\documentclass{article}

\pdfoutput=1
\usepackage[english]{babel}
\usepackage[letterpaper,top=2cm,bottom=2cm,left=3cm,right=3cm,marginparwidth=1.75cm]{geometry}

\usepackage[utf8]{inputenc} 
\usepackage[T1]{fontenc}    
\usepackage{hyperref}       
\usepackage{url}            
\usepackage{booktabs}       
\usepackage{amsfonts}       
\usepackage{nicefrac}       
\usepackage{microtype}      
\usepackage{xcolor}         

\usepackage{graphicx}
\usepackage{makecell}
\usepackage{caption}
\usepackage{subcaption}
\usepackage{multirow}

\usepackage{colortbl}
\definecolor{bgcolor}{rgb}{0.66,0.88,1.00}

\usepackage{amsfonts,amsmath,amsthm,amssymb,dsfont,etex}
\allowdisplaybreaks
\usepackage{mathrsfs}
\usepackage{amsmath}
\usepackage{adjustbox}

\theoremstyle{plain}
\newtheorem{theorem}{Theorem}[section]

\newtheorem{lemma}[theorem]{Lemma}
\newtheorem{fact}[theorem]{Fact}

\newtheorem{corollary}[theorem]{Corollary}
\newtheorem{claim}[theorem]{Claim}
\newtheorem{problem}[theorem]{Problem}
\theoremstyle{definition}
\newtheorem{definition}[theorem]{Definition}

\usepackage{threeparttable}

\usepackage{xspace}

\xspaceaddexceptions{[]\{\}}

\usepackage{hyperref}

\usepackage{thm-restate}

\usepackage[linesnumbered,ruled,vlined]{algorithm2e}
\SetKwComment{Comment}{$\triangleright$\ }{}

\newcommand{\eat}[1]{}

\makeatletter
\newcommand*{\rom}[1]{\expandafter\@slowromancap\romannumeral #1@}
\makeatother

\title{Efficient Submodular Optimization under Noise: Local Search is Robust}
\author{Lingxiao Huang\thanks{{\tt huanglingxiao1990@126.com}. Huawei TCS Lab $\to$ Nanjing University.}
  \and
  Yuyi Wang \thanks{{\tt yuyiwang920@gmail.com}. Swiss Federal Institute of Technology.}
  \and
  Chunxue Yang \thanks{{\tt chunxue001@e.ntu.edu.sg}. Nanyang Technological University.}
  \and 
  Huanjian Zhou\thanks{{\tt zhou@ms.k.u-tokyo.ac.jp}. The University of Tokyo.}
  }
\begin{document}
\maketitle

\begin{abstract}
The problem of monotone submodular maximization has been studied extensively due to its wide range of applications.
However, there are cases where one can only access the objective function in a distorted or noisy form because of the uncertain nature or the errors involved in the evaluation.
This paper considers the problem of constrained monotone submodular maximization with noisy oracles introduced by~\cite{hassidim2017submodular}.
For a cardinality constraint, we propose an algorithm achieving a near-optimal $\left(1-\frac{1}{e}-O(\varepsilon)\right)$-approximation guarantee (for arbitrary $\varepsilon > 0$) with only a polynomial number of queries to the noisy value oracle, which improves the exponential query complexity of~\cite{singer2018optimization}.
For general matroid constraints, we show the first constant approximation algorithm in the presence of noise.
Our main approaches are to design a novel local search framework that can handle the effect of noise and to construct certain smoothing surrogate functions for noise reduction.
\end{abstract}
\newpage
\tableofcontents
\newpage
	
\section{Introduction} 
\label{sec:introduction}
Consider the following problems in machine learning and operations research: (1) selecting a set of locations to open up facilities with the goal of maximizing their overall user coverage~\cite{krause2008efficient}; (2) reducing the number of features in a machine learning model while retaining as much information as possible~\cite{thoma2009near}; and (3) identifying a small set of seed nodes that can achieve the largest overall influence in a social network~\cite{kempe2003maximizing}.
Solving these problems all involves maximizing a monotone \emph{submodular} set function $f: 2^N \mapsto \mathbb{R}$ subject to certain constraints. 
Intuitively, submodularity captures the property of diminishing returns. For example, a newly opened facility will contribute less to the overall user coverage if we have already opened many facilities and more if we have only opened a few. 
Although the general problem of monotone submodular maximization subject to a cardinality or general matroid constraint is NP-hard~\cite{feige1998threshold}, 
the greedy algorithm, which selects an element with the largest margin at each step, can approximately solve this problem under a cardinality constraint by a factor of $1-1/e$, and this approximation ratio is tight~\cite{nemhauser1978}. 
Moreover, a non-oblivious local search algorithm, which iteratively alters one element to improve an auxiliary objective function, is guaranteed to achieve an approximation ratio of $1-1/e$ for general matroid constraints~\cite{filmus2014monotone}.

In the literature, the submodular optimization problem usually assumes a \emph{value oracle} to the objective function $f$, which means one is allowed to query the \emph{exact} value of $f(S)$ for any $S \subseteq N$. 
However, in many applications, due to the uncertain nature of the objective or the errors involved in the evaluation, one can only access the function value in a distorted or noisy form.
For example, \cite{globerson2006nightmare} pointed out that selecting features to construct a robust learning model is particularly important in domains with non-stationary feature distributions or input sensor failures. 
It is known that without any assumption on the noise, direct adaptions of the greedy and local search methods mentioned above may yield arbitrarily poor performance~\cite{horel2016maximization}.
To address this issue, \cite{hassidim2017submodular} introduced and studied the following problem of \emph{monotone submodular maximization under noise}:
For a monotone submodular function $f$, given a noisy value oracle $\Tilde{f}$ satisfying $\Tilde{f}(S) = \xi_Sf(S)$ for each set $S\subseteq N$, where the noise multiplier $\xi_S$ is independently drawn from a certain distribution,
the goal is to find a set $S$ maximizing $f(S)$ under certain constraints. 
Some applications of this problem are provided in~\cite{singer2018optimization} such as revealed preference theory~\cite{chambers2016revealed} and active learning~\cite{feldman2009power}.
\cite{hassidim2017submodular} showed that under a sufficiently large cardinality constraint, a variant of the greedy algorithm achieves a near-optimal approximation ratio of $1-\frac{1}{e}-O(\varepsilon)$ for arbitrary $\varepsilon > 0$. 
The problem becomes more challenging when the cardinality constraint is relatively small since there is less room for mistakes. 
In a subsequent work, \cite{singer2018optimization} developed another greedy-based algorithm for small cardinality constraints and showed a near-tight approximation guarantee. 

Despite these encouraging results, we want to point out two directions along this line of monotone submodular maximization under noise that still have room for improvements:
\begin{itemize}
    \item In the algorithm from ~\cite{singer2018optimization}, the query complexity to the noisy value oracle is exponential in $\varepsilon^{-1}$, which is costly when a near-optimal solution is needed, i.e., when the parameter $\varepsilon$ is close to $0$. Is it possible to obtain near-optimal approximations for monotone submodular maximization under cardinality constraints with the number of queries polynomial in $\varepsilon^{-1}$?
    \item All previous works in submodular maximization under noise consider only the cardinality constraint, and no approximation guarantee is known under other constraints.
    \footnote{Note that~\cite{singer2018optimization} also provide an algorithm to deal with general matroid constraints.
    However, we will argue in Section~\ref{sec:invalid} that their algorithm fails to obtain the approximation guarantee they claim.}
    Is there any algorithm that can achieve a constant approximation for submodular maximization under noise for more general constraints, such as commonly studied    matroid constraints~\cite{krause2009simultaneous,vondrak2011submodular}?
\end{itemize}
In this paper, we provide answers to both questions above. 
\subsection{Our contributions}
\label{sec:contribution}
We study the problem of constrained monotone submodular maximization under noise (Problem~\ref{problem:noisy_submodular}).
Following prior works~\cite{hassidim2017submodular,singer2018optimization}, we assume generalized exponential tail noises (Definition~\ref{def:exponential}) and consider the solutions subject to cardinality constraints (Definition~\ref{def:cardinality}) and matroid constraints (Definition~\ref{def:matroid}).
The main contribution of this work is to show that for optimizing a monotone submodular function under a cardinality constraint, $\left(1-\frac{1}{e}-O(\varepsilon)\right)$-approximations can be obtained with high probability by querying the noisy oracle only $\text{Poly}\left(n,\frac{1}{\varepsilon}\right)$ times.
\begin{theorem}[\textbf{Informal, see Theorems~\ref{thm:cardi} and~\ref{thm:cardi2}}]
\label{thm:informal_cardi}
    Let $\varepsilon > 0$ and assume $n$ is sufficiently large. 
    For any $r\in \Omega\left(\frac{1}{\varepsilon}\right)$, there exists an algorithm that returns a $\left(1-\frac{1}{e}-O(\varepsilon)\right)$-approximation for the monotone submodular maximization problem under a $r$-cardinality constraint, with probability $1-o(1)$ and query complexity $\text{Poly}\left(n,\frac{1}{\varepsilon}\right)$ to $\tilde{f}$.
\end{theorem}
\noindent
For a cardinality constraint, this paper and prior works~\cite{hassidim2017submodular, singer2018optimization} 
all achieve near-optimal approximations.
However, our result is applicable for a larger range $\Omega\left(\frac{1}{\varepsilon}\right)$ of cardinalities than $\Omega\left(\log\log n\cdot\varepsilon^{-2}\right)$ in~\cite{hassidim2017submodular}.
Moreover, we only require $\text{Poly}\left(n,\frac{1}{\varepsilon}\right)$ queries to the noisy value oracle, 
which improves the query complexity $\Omega(n^{\frac{1}{\varepsilon}})$ of~\cite{singer2018optimization}.
Our main idea to address this problem is to employ a local search procedure, whereas prior methods are all variants of the greedy algorithm.
Intuitively, local search is more robust than the greedy since the marginal functions are more sensitive to noise than the value functions.
To achieve a sufficient gain in each iteration, the greedy algorithm must identify the element with maximum margin.
In contrast, local search only needs to estimate the sets' values.
This is why we can improve the query complexity to $\text{Poly}\left(n, \frac{1}{\varepsilon}\right)$ for small cardinality constraints.

We present a unified framework (Algorithm~\ref{alg:local_search}) for enhancing the local search to cope with noise.
One of the main differences between our framework and the non-oblivious local search proposed by~\cite{filmus2014monotone} is that we use an approximation of the auxiliary function (Definition~\ref{def:approximation})
rather than the exact one due to the presence of noise.
We analyze the impact of the inaccuracies on the approximation performance and query complexity of the local search (Theorem~\ref{the:main_nls}).
Another difference is the construction of the auxiliary functions used to guide the local search.
We construct the auxiliary function not based on the objective function but on smoothing surrogate functions.
A surrogate function $h$ needs to meet two properties: (\romannumeral1) $h$ should depend on an averaging set of size $\mathrm{poly}(n)$, such that $h(S)$ and its noisy analogue $\tilde{h}(S)$ are close for all sets $S$ considered by local search; (\romannumeral2) $h$ needs to be close to the original function $f$, such that optimizing $h$ can yield a near-optimal solution to optimizing $f$.
However, a large averaging set is more likely to induce a large gap between the surrogate and original function, making simultaneous fulfillment of both properties non-trivial. 
In this paper we carefully design the smoothing surrogate functions as follows.
For a set with size $r\in\Omega\left(\frac{1}{\varepsilon}\right) \cap O\left(n^{1/3}\right)$, we define the smoothing surrogate function $h$ as the expectation of $f$'s value when a random element in $N$ is added to the set (Definition~\ref{def:h_1}). 
This surrogate function is robust for a relatively small cardinality as it is based on a rather large averaging set with size nearly $n$, but too concentrated for a large cardinality close to $n$.
Thus, we consider another smoothing surrogate function $h_H$ for size $r\in \Omega\left(n^{1/3}\right)$, defined as the average value combined with all subsets of a certain small-size set $H$ (Definition~\ref{def:h_2}).
The auxiliary functions constructed on both smoothing surrogates are shown to have almost accurate approximations (Lemma~\ref{lem:hat_varphi_h} and~\ref{lem:hat_varphi_h2}).
Consequently, we can apply our unified local search framework (Algorithm~\ref{alg:local_search}) in both cases (Algorithm~\ref{alg:local_search_monotone} and~\ref{alg:local_search_monotone2}),
and guarantee to achieve nearly tight approximate solutions (Theorem~\ref{thm:cardi} and~\ref{thm:cardi2}).
The other contribution of this paper is a constant approximation result for maximizing monotone submodular functions with noisy oracles under general matroid constraints.
\begin{theorem}[\textbf{Informal, see Theorems~\ref{thm:matroid1} and~\ref{thm:matroid2}}]
    Let $\varepsilon > 0$ and assume $n$ is sufficiently large. 
    For any $r\in \Omega\left(\varepsilon^{-1}\log(\varepsilon^{-1})\right)$, there exists an algorithm that returns a $\left(\frac{1}{2}\left(1-\frac{1}{e}\right)-O(\varepsilon)\right)$-approximation for the monotone submodular maximization problem under a matroid constraint with rank $r$, with probability $1-o(1)$ and query complexity at most $\text{Poly}\left(n,\frac{1}{\varepsilon}\right)$ to $\tilde{f}$.
\end{theorem}
\noindent
To the best of our knowledge, this is the first result showing that constant approximation guarantees are obtainable 
under general matroid constraints in the presence of noise.
To cope with noise, one common approach for cardinality constraints is to incorporate some extra elements to gain robustness and include these elements in the final solutions.
However, for a matroid, additional elements may undermine the independence of a set.
To address this difficulty, we develop a technique for comparing the values of independent sets in the presence of noise, which allows us to select either the local search solutions or the additional elements for robustness and leads to an approximation ratio of $\frac{1}{2}\left(1-\frac{1}{e}\right)$. 
\subsection{Related work}
\label{sec:related}
Research has been conducted on monotone submodular maximization in the presence of noise.
We say a noisy oracle is inconsistent if it returns different answers when repeatedly queried.
For inconsistent oracles, noise often does not present a barrier to optimization, since concentration assumptions can eliminate the noise after a sufficient number of queries~\cite{singla2016noisy,ito2019submodular}.
When identical queries always obtain the same answer, the problem becomes more challenging.
Aside from the i.i.d noise adopted in~\cite{hassidim2017submodular,singer2018optimization} and this paper, \cite{horel2016maximization} study submodular optimization under noise adversarially generated from $[1-\varepsilon/r, 1+\varepsilon/r]$, where the greedy algorithm achieves a ratio of $1-1/e-O(\varepsilon)$.
No algorithm can obtain a constant approximation if noise is not bounded in this range.

\section{The model}
\label{sec:model}
This section formally defines our model (Problem~\ref{problem:noisy_submodular}) of maximizing a monotone submodular function (Definition~\ref{def:submodular}) under a cardinality constraint (Definition~\ref{def:cardinality}) or a matroid constraint (Definition~\ref{def:matroid}), given access to a noisy value oracle (Definition~\ref{def:noise}).
Let $N$ be the ground set with size $|N|=n$, and we use the shorthands $S+x = S\cup\{x\}$ and $S-x=S\backslash\{x\}$ throughout this paper. 
We first review the definition of monotone submodular functions. 
\begin{definition}[\bf{Monontone submodular functions}]
\label{def:submodular}
A function $f :2^N \to \mathbb{R}_{\geq 0}$
is monotone submodular if 1) (monotonicity) $f(A)\leq f(B)$ for any $A\subseteq B\subseteq N$;
 2) (submodularity) for any subset $A,B \subseteq N$ and $x\in N$:
$f(A + x) - f(A)\geq  f(B + x) - f(B)$.
\end{definition}
\noindent 
There are various examples of monotone submodular functions in optimization, such as budget additive functions, coverage functions, cut functions and rank functions~\cite{buchbinder2018submodular}.
Since the description of a submodular function may be exponential in the size $N$, we usually assume access of a \emph{value oracle} that answers $f(S)$ for each $S\subseteq N$. 
Let $\mathcal{I}$ denote a collection of feasible subsets $S\subseteq N$. 
The goal of a constrained monotone submodular maximization is to find a subset $S\subseteq \mathcal{I}$ to maximize $f(S)$. 
The objective $f$ is further assumed to be normalized, i.e., $f(\varnothing) = 0$.
We consider two types of constraints in our models: cardinality constraints and matroid constraints. 
\begin{definition}[\bf{Cardinality constraints}]
\label{def:cardinality}
Given a ground set $N$ and an integer $r\geq 1$, a cardinality constraint is of the form $\mathcal{I}(r) = \{ S\subseteq N: |S|\leq r\}$.
\end{definition}
\begin{definition}[\bf{Matroids and Matroid constraints}]
\label{def:matroid}
Given a ground set $N$, a matroid $\mathcal{M}$ is represented by an ordered pair $(N,\mathcal{I}(\mathcal{M}))$ satisfying that
1) $\varnothing \in \mathcal{I}(\mathcal{M})$; 2) If $I\in \mathcal{I}(\mathcal{M})$ and $I'\subseteq I$, then $I'\in \mathcal{I}(\mathcal{M})$; 3) If $I_1,I_2\subseteq \mathcal{I}(\mathcal{M})$ and $|I_1|<|I_2|$, then there must exist an element $e\in I_2\setminus I_1$ such that $I_1\cup \{e\}\in \mathcal{I}(\mathcal{M})$.
Each $I\in \mathcal{I}(\mathcal{M})$ is called an independent set. 
The maximum size of an independent set is called the rank of $\mathcal{M}$.
We call the collection $\mathcal{I}(\mathcal{M})$ a matroid constraint.
\end{definition}

\noindent 
Assume we are given a \emph{membership oracle} of $\mathcal{I}(\mathcal{M})$ that for any set $S \subseteq N$ answers whether $S\in \mathcal{I}(\mathcal{M})$. 
As a widely used combinatorial structure, there is an extensive study on matroids~\cite{wilson1973introduction,oxley2006matroid,welsh2010matroid}.
Common matroids include uniform matroids, partition matroids, regular matroids, etc. See~\cite{oxley2006matroid} for more discussions. 
Specifically, a uniform matroid constraint is equivalent to a cardinality constraint, implying that cardinality constraints are a special case of matroid constraints.
\paragraph{Noisy value oracle.} 
It is well known that given an exact value oracle to $f$, for any $\varepsilon > 0$, there exists a randomized $(1-1/e-\varepsilon)$-approximate algorithm for the submodular optimization problem under a matroid constraint, i.e., $\max_{S\subseteq \mathcal{I}}f(S)$~\cite{calinescu2011maximizing,filmus2014monotone}.
However, as discussed earlier, the value oracle of $f$ may be imperfect, and we may only have a noisy value oracle $\tilde{f}$ instead of $f$. 
We consider the following noisy value oracle that has also been investigated in~\cite{horel2016maximization,hassidim2017submodular}.

\begin{definition}[\bf{(Multiplicative) noisy value oracle}~\cite{hassidim2017submodular}]
\label{def:noise}
We call $\Tilde{f} : 2^N \to \mathbb{R}_{\geq 0}$ a (multiplicative) noisy value oracle of $f$, if there exists some distribution $\mathcal{D}$ s.t.\ for any $S \subseteq N$, $\tilde{f}(S) = \xi_Sf(S)$ where $\xi_S$ is i.i.d.\ drawn from $\mathcal{D}$.
\end{definition}

\noindent
Throughout this paper, we consider a general class of noise distributions, called generalized exponential tail distributions, defined in \cite{singer2018optimization}.
\begin{definition}[\bf{Generalized exponential tail distributions~\cite{singer2018optimization}}]
\label{def:exponential}
A noise distribution $\mathcal{D}$ has a \textit{generalized exponential tail} if there exists some $x_0$ such that for every $x > x_0$ the probability density function $\rho(x) = e^{-g(x)}$, where $g(x) = \sum_i c_ix^{\gamma_i}$ for some (not necessarily integers) $\gamma_0 \geq \gamma_1 \geq \dots, s.t. \gamma_0 \geq 1$ and $c_0 > 0$. 
If $\mathcal{D}$ has bounded support we only require that either it has an atom at its supremum or that $\rho$ is continuous and non-zero at the supremum.
\end{definition}

\noindent
W.l.o.g., we assume $\mathbb{E}[\mathcal{D}]=1$.\footnote{Otherwise, we can scale $\mathcal{D}$ to be $\mathcal{D}'= \mathcal{D}/\mathbb{E}[\mathcal{D}]=1$.}
As mentioned in~\cite{singer2018optimization}, the class of generalized exponential tail distributions contains Gaussian distributions, exponential distributions, and all distributions with bounded support that is independent of $n$. 

\paragraph{The model.} We are ready to propose the main problem, which has already been considered in~\cite{hassidim2017submodular,singer2018optimization}.
\begin{problem}[\bf{Constrained submodular optimization under noise}]
\label{problem:noisy_submodular}
Given a noisy value oracle $\Tilde{f}: 2^N\to \mathbb{R}_{\geq 0}$ (with a certain generalized exponential tail distribution $\mathcal{D}$) to an underlying monotone submodular function $f$, and a cardinality constraint $\mathcal{I}=\mathcal{I}(r)$ or a matroid constraint $\mathcal{I}=\mathcal{I}(\mathcal{M})$, the goal is to find $S\subseteq \mathcal{I}$ to maximize $f(S)$. 
\end{problem}

\section{Local search with approximate evaluation oracles to auxiliary functions}
\label{sec:meta_algorithm}

This section proposes a non-oblivious local search framework with a noisy value oracle (Algorithm~\ref{alg:local_search}), and gives an analysis (Theorem~\ref{the:main_nls}) of its performance.
This algorithm is a generalization of that in~\cite{filmus2014monotone} and will be used as a meta-algorithm for solving Problem~\ref{problem:noisy_submodular} in Sections~\ref{sec:cardinality} and~\ref{sec:matroid}.
Roughly speaking, a local search framework first constructs a "good" initial solution via a standard greedy algorithm and then iteratively improves the solution w.r.t.\ to an auxiliary function (Definition~\ref{def:auxiliary}) by swapping an element at each step.
The following auxiliary function proposed by \cite{filmus2014monotone} is a linear combination of the original function over all subsets. 
%
\begin{definition}[\bf{Auxiliary function~\cite{filmus2014monotone}}]
\label{def:auxiliary}
Given a monotone submodular function $h:2^N\to \mathbb{R}_{\geq 0}$, its auxiliary function $\varphi_h:2^N\to \mathbb{R}_{\geq 0}$ is defined as
$\varphi_h(S) = \sum_{T\subseteq S} m_{|S|-1,|T|-1} \cdot h(T)$,
where
$m_{s, t} = 
     \int_0^1 e^p(e-1)^{-1} p^t(1-p)^{s-t}dp$
for all $s\geq t\geq 0$. 
\end{definition}
\noindent
We use $h$ instead of $f$ because we may do local search on certain smoothing surrogate function $h$ of $f$ (see examples in Definitions~\ref{def:h_1} and~\ref{def:h_2}). 
Note that $\varphi_h$ is also a monotone submodular function.
Given an exact value oracle of $\varphi_h$ and a matroid constraint $\mathcal{I}(\mathcal{M})$, \cite{filmus2014monotone} show that there exists a local search algorithm that outputs a $(1-1/e)$-approximate solution for the optimization problem $\max_{S\subseteq \mathcal{I}(\mathcal{M})}h(S)$.
However, as mentioned before, we do not have an exact value oracle of $h$. 
Thus, it may not be possible to construct an exact value oracle of $\varphi_h$.
To understand how these inaccuracies can affect the performance of the local search algorithm, we introduce the following approximation for the auxiliary function $\varphi_h$.
\begin{definition}[\bf{$(\alpha,\delta,\mathcal{I}(\mathcal{M}))$-approximation of $\varphi_h$}]
\label{def:approximation}

Given an auxiliary function $\varphi_h$, a matroid constraint $\mathcal{I}(\mathcal{M})$ and constants $\alpha, \delta > 0$, we say a randomized function $\widehat{\varphi_h}$ is an $(\alpha,\delta,\mathcal{I}(\mathcal{M}))$-approximation of $\varphi_h$ if for any independent set $A\in \mathcal{I}(\mathcal{M})$, 
\begin{equation*}
    \mathbb{P}\left[\left|\widehat{\varphi_h}(A) - \varphi_h(A)\right| > \alpha\cdot \max_{S\in \mathcal{I}(\mathcal{M})} \varphi_h(S)\right]\leq \delta.
\end{equation*}
\end{definition}

\noindent
Intuitively, $\widehat{\varphi_h}$ is a randomized version of $\varphi_h$ with an additive concentration guarantee, where the randomness may come from noise of $h$ or sampling error for estimating $\varphi_h$.
As $\alpha, \delta$ tend to 0, $\widehat{\varphi_h}$ is a better approximation of $\varphi_h$.
Specifically, when $\alpha = \delta = 0$ and $\mathcal{I}=2^N$, $\widehat{\varphi_h}$ is exactly equivalent to $\varphi_h$.
Now we are ready to provide our local search framework (Algorithm~\ref{alg:local_search}), which is a generalization of \cite[Algorithm 2]{filmus2014monotone}. 
The main difference is that we use an approximation $\widehat{\varphi_h}$ instead of $\varphi_h$ for evaluation at each greedy step (Line 4) and local search step (Line 12).
The performance of Algorithm~\ref{alg:local_search} is summarized in Theorem~\ref{the:main_nls}.
\begin{algorithm}[htp]
 \caption{Noisy local search (\texttt{NLS}($\widehat{\varphi_h},\mathcal{I}(\mathcal{M}),\Delta$))}
  \label{alg:local_search}
\DontPrintSemicolon
\SetNoFillComment
  \SetKwInOut{Input}{Input}
  \Input{ A matroid constraint $\mathcal{I}(\mathcal{M})$ of rank $r\geq 1$, a value oracle to an $(\alpha,\delta,\mathcal{I}(\mathcal{M}))$-approximation $\widehat{\varphi_h}$ of $\varphi_h$,
		a stepsize $\Delta\in (0,1/2)$.}
    Set $I \leftarrow \left\lceil\log_{1+\Delta}\left(\frac{2(1+\alpha)}{1-2(r+1)\alpha}\right) \right\rceil$\;
    Initialize: $U_0 \leftarrow \varnothing$, $i = 1$\;
    \While{$i\leq r$}{
     $u_i \leftarrow \mathop{\arg\max}_{e: U_{i-1}+e\in \mathcal{I}(\mathcal{M})}\widehat{\varphi_h}(U_{i-1}+e)$\;
    $U_{i}\leftarrow U_{i-1}+ u_i$\;
     $i\leftarrow i+1$\;
    }
    $S_0 \leftarrow U_{r}$ \Comment*[r]{Initial solution by greedy}
    \For(\Comment*[f]{Local search for  $I$ iterations}){$i =0$ to $I-1$}{
        \For{each element $x \in S_i$ and $y \in N \setminus S_i$}{
            $S'_i \leftarrow S_i-x+y$\;
            \If{$S_i'\in \mathcal{I}(\mathcal{M})$}{
                \If(\Comment*[f]{An improved solution $S_i'$ was found}){$\widehat{\varphi_h}(S'_i)\geq (1+\Delta)\cdot \widehat{\varphi_h}(S_i)$}{
                $S_{i+1}\leftarrow S'_i$\; 
                break and continue to the next iteration of $i$ 
                }
            }
        }
    }
\Return $S_I$
\end{algorithm}
\begin{theorem}[\textbf{Performance of noisy local search}]
\label{the:main_nls}
Let $I = \log_{1+\Delta}\left(\frac{2(1+\alpha)}{1-2(r+1)\alpha}\right)$.
With probability at least $1-(I+1)rn\delta$, the output $S_I$ of Algorithm \ref{alg:local_search} is a $\left(1-\frac{1}{e}\right)\left(1-r\ln (er)\left((2+\Delta)\alpha+\Delta\right)\right) $-approximate solution for problem $\max_{S\subseteq \mathcal{I}(\mathcal{M})}h(S)$,  with at most  $\left(I+1\right)rn$ calls to $\widehat{\varphi_h}$.
\end{theorem}
\noindent
Roughly, given a nearly accurate approximation $\widehat{\varphi_h}$, we can show that the initial solution $S_0$ is likely to be a constant approximation for $\max_{S\in \mathcal{I}(\mathcal{M})}\varphi_h(S)$ (Lemma \ref{lmm:init_solution}). 
This guarantees that we will reach a local optimal solution w.r.t. $\varphi_h$ in $I$ iterations with high probability (Claim~\ref{clm:App_A}).
Then combining with basic properties of $\varphi_h$ shown in~\cite{filmus2014monotone} and the fact that $\widehat{\varphi_h}$ is "close" to $\varphi_h$, we obtain Theorem~\ref{the:main_nls}.

\begin{proof}
We first analyze the query complexity to $\widehat{\varphi_h}$ and then prove the approximate ratio together with the success probability.

\paragraph{Query complexity of Algorithm \ref{alg:local_search}.}
For any $i\in [r]$, define $\mathcal{U}_i := \{U_{i-1}+e\in \mathcal{I}(\mathcal{M}): e\in N\backslash U_{i-1}\}$ to be the collection of subsets considered in Line 4. For any $i = 0,1\cdots ,I-1$, define $\mathcal{U}_i' := \{S_i'\in \mathcal{I}(\mathcal{M})\}$ to be the collection of subsets considered in Line 10 in interation $i$. 
Let $\mathcal{U} = \left(\mathop{\bigcup}\limits_{i\in [r]}\mathcal{U}_i\right) \cup \left(\mathop{\bigcup}\limits_{0\leq i\leq I-1}\mathcal{U}_i'\right)$. Since each $\left|\mathcal{U}_i\right|\leq n$ and each $\left|\mathcal{U}_i'\right|\leq rn$, we have 
\[\left|\mathcal{U}\right|\leq (I+1)rn.\]
Thus the total query complexity is less than $(I+1)rn$.

\paragraph{Approximate ratio and success probability of Algorithm \ref{alg:local_search}.}
We say that event $\mathcal{E}_g$ happens if for all sets $A\in \mathcal{U}$, the following inequality holds:
\[ \mathbb{P}\left[\left|\widehat{\varphi_h}(A) - \varphi_h(A)\right| > \alpha \max_{S\in \mathcal{I}(\mathcal{M})} \varphi_h(S)\right]\leq \delta.\]

By the above argument, we know that the total number of calls to $\widehat{\varphi_h}$ is at most $(I + 1)rn$. Thus \[\mathbb{P}[\mathcal{E}_g]\geq 1-(I+1)rn\delta.\]
In the remaining proof, we assume $\mathcal{E}_g$ holds. Firstly we prove the following lemma which analyzes the performance of the greedy initial solution $S_0$ in Line 7 of Algorithm~\ref{alg:local_search}.
\begin{lemma}[\bf{Approximation ratio of $S_0$}]
\label{lmm:init_solution}
Conditioned on $\mathcal{E}_g$, we have
\[
\varphi_h(S_0)\geq \frac{1-2r\alpha}{2}\max_{S\in \mathcal{I}(\mathcal{M})}\varphi_h(S).
\]
\end{lemma}
\begin{proof}

We first introduce the following well-known fact.
\begin{fact}[\textbf{Brualdi’s lemma}]
\label{fact:Brualdi}
Suppose $A$, $B$ are two bases in a matroid. There is a bijection $\pi : A \to B$ such that for all $a \in  A$, $A - a + \pi(a)$ is a base. Furthermore, $\pi$ is the identity on $A \cap B$.
\end{fact}
\noindent
Let $O^\star\in \mathop{\arg\max}\limits_{S\in\mathcal{I}(\mathcal{M})}\varphi_h(S)$ be a set on which $\varphi_h$ attains its maximum value.
For $O^\star$ and $U_r$,
we index the elements of $O^\star$ as $\{o_1, \dots, o_{r}\}$ according to the Brualdi's lemma such that $o_i = \pi(u_i)$ with
$U_{i-1}+o_i\in \mathcal{I}(\mathcal{M})$.
From how we choose $u_i$ in Line 4 of Algorithm~\ref{alg:local_search}, we have
\begin{equation}
\label{eq:optimal_greedy}
    \widehat{\varphi_h}(U_{i-1}+u_i) \geq \widehat{\varphi_h}(U_{i-1}+o_i)
\end{equation}
for all $i \in [r]$.
To utilize monotonicity and submodularity of $\varphi_h$, we translate the inequality above into that of $\varphi_h$:
\begin{align*}
    \varphi_h(U_{i-1}+u_i)+\alpha\varphi_h(O^\star) &~\geq~  \widehat{\varphi_h}(U_{i-1}+u_i) && \text{(event $\mathcal{E}_g$)}\\
    &~\geq~ \widehat{\varphi_h}(U_{i-1}+o_i) && \text{(Ineq.~\eqref{eq:optimal_greedy})}\\ &~\geq~ \varphi_h(U_{i-1}+o_i)-\alpha\varphi_h(O^\star) && \text{(event $\mathcal{E}_g$)}.
\end{align*}
That is,
$$\varphi_h(U_{i-1}+u_i) \geq \varphi_h(U_{i-1}+o_i)- 2\alpha\cdot\varphi_h(O^\star).$$
Subtracting $\varphi_h(U_{i-1})$ from each side gives
$$\varphi_h(U_{i-1} + u_i) - \varphi_h(U_{i-1}) \geq \left[\varphi_h(U_{i-1} + o_i) - \varphi_h(U_{i-1})\right] - 2\alpha\cdot\varphi_h(O^\star)$$
for all $i\in [r]$.
Summing these $r$ inequalities, we obtain
\begin{equation}
\label{eq:greedy_1}
    \varphi_h(S_0) \geq \sum_{i=1}^{r-1} [\varphi_h(U_{i-1}+o_i) - \varphi_h(U_{i-1})] - 2r\alpha\cdot\varphi_h(O^\star).
\end{equation}
From submodularity of $\varphi_h$, we have
\begin{equation}
\label{eq:greedy_2}
    \sum_{i=1}^{r-1} [\varphi_h(U_{i-1}+o_i) - \varphi_h(U_{i-1})] \geq \sum_{i=1}^{r-1} [\varphi_h(S_0+o_i) - \varphi_h(S_0)] \geq \varphi_h(S_0\cup O^\star)-\varphi_h(S_0).
\end{equation}
Combining~\eqref{eq:greedy_1} and~\eqref{eq:greedy_2} gives
$$2\varphi_h(S_0) \geq
\varphi_h(S_0\cup O^\star) - 2r\alpha\cdot\varphi_h(O^\star)\geq (1-2r\alpha)\varphi_h(O^\star).$$
This completes the proof.
\end{proof}

\noindent
\begin{lemma}[\textbf{Maximum iterations of the local search procedure in Algorithm~\ref{alg:local_search}}]
\label{clm:App_A}
Conditioned on event $\mathcal{E}_g$, the local search procedure in Algorithm~\ref{alg:local_search} terminates within $I = \log_{1+\Delta}\left(\frac{2(1+\alpha)}{1-2(r+1)\alpha}\right)$ iterations.
\end{lemma}

\begin{proof}
Let $O^\star\in \mathop{\arg\max}\limits_{S\in\mathcal{I}(\mathcal{M})}\varphi_h(S)$ be a set on which $\varphi_h$ attains its maximum value and $O'\in \mathop{\arg\max}\limits_{S\in\mathcal{U}}\widehat{\varphi_h}(S)$ be the set with largest value of $\widehat{\varphi_h}$.
By Lemma \ref{lmm:init_solution}, we have
\begin{align*}
\widehat{\varphi_h}(S_0)
~\geq~ &\varphi_h(S_0)- \alpha \varphi_h(O^\star)\\
\geq~&\frac{1-2(r+1)\alpha}{2} \varphi_h(O^*)\tag*{(Lemma \ref{lmm:init_solution})}\\
\geq~ &\frac{1-2(r+1)\alpha}{2(1+\alpha)} \varphi_h(O')
+ \frac{(1-2(r+1)\alpha)\alpha}{2(1+\alpha)} \varphi_h(O^*)
\tag*{(optimality of $O^*$)}\\
\geq~& \frac{1-2(r+1)\alpha}{2(1+\alpha)} \left(\widehat{\varphi_h}(O') - \alpha \varphi_h(O^*)\right) + \frac{(1-2(r+1)\alpha)\alpha}{2(1+\alpha)} \varphi_h(O^*) \\
=~&  \frac{1-2(r+1)\alpha}{2(1+\alpha)} \widehat{\varphi_h}(O').
\end{align*}
Since every local search step increases the function value $\widehat{\varphi_h}$ by $1+\Delta$, Algorithm~\ref{alg:local_search}  searches for at most $I =  \log_{1+\Delta}\left(\frac{2(1+\alpha)}{1-2(r+1)\alpha}\right)$ iterations.
\end{proof}
\noindent
Conditioned on $\mathcal{E}_g$, we finally obtain the approximation ratio of $S_I$ by the locally optimality of $S_I$ by Theorem 1 and Lemma 4 in \cite{filmus2014monotone}.
\begin{lemma}[\textbf{\cite[Theorem 1]{filmus2014monotone}}] 
\label{lmm:local_sub}
Let $A = \{a_1,\cdots,a_r\}$ and $B = \{b_1,\cdots,b_r\}$ be any two bases of $\mathcal{M}$. Further suppose that we index the elements of $B$ so that $b_i = \pi(a_i)$, where $\pi : A \to B$ is the bijection guaranteed by Fact~\ref{fact:Brualdi} (Brualdi’s lemma).
Then,
\[ \sum\limits_{i=1}^r\left[\varphi_h(A-a_i+b_i)-\varphi_h(A)\right]\geq h(B)-\frac{e}{e-1}h(A).\]
\end{lemma}
\begin{lemma}[\textbf{\cite[Lemma 4]{filmus2014monotone}}]
\label{lmm:bound_g}
For all $A\subseteq N$,
\[h(A) 
\leq \varphi_h(A) 
\leq  \frac{e}{e-1}H_{|A|}h(A),\]
where $H_{|A|} = \sum\limits_{i=1}^{|A|}\frac{1}{i}.$
\end{lemma}
\noindent
Let $O \in \mathop{\arg\max}\limits_{S\in\mathcal{I}(\mathcal{M})}h(S)$ and  $O^\star\in \mathop{\arg\max}\limits_{S\in\mathcal{I}(\mathcal{M})}\varphi_h(S)$. By Fact~\ref{fact:Brualdi}, we index the elements of $S_I$ and $O$ by $S_I=\{s_1,\cdots,s_r\}$ and $O =\{o_1,\cdots,o_r\} $ such that $S_I-s_i+o_i\in \mathcal{I}(\mathcal{M})$ holds for all $i\in [r]$.
We firstly transfer $S_I$'s the locally optimality w.r.t. $\widehat{\varphi_h}$ to the locally optimality w.r.t. $\varphi_h$.
For all $s_i$ and $o_i$, we have 
\begin{align*}
    \varphi_h(S_I-s_i+o_i)- \alpha \varphi_h(O^*) 
    &~\leq~ \widehat{\varphi_h}(S_I-s_i+o_i)\tag*{(event $\mathcal{E}_g$)}\\
    &~\leq~ (1+\Delta) \widehat{\varphi_h}(S_I)\tag*{(Lemma \ref{clm:App_A})}\\
    &~\leq~(1+\Delta)\left(\varphi_h(S_I)+  \alpha \varphi_h(O^\star)\right)\tag*{(event $\mathcal{E}_g$)}
\end{align*}

which implies
\[\varphi_h(S_I-s_i+o_i) -\varphi_h(S_I) \leq \Delta \varphi_h(S_I) + (2+\Delta)\alpha \varphi_h(O^\star).\]
Summing the resulting $r$ inequalities gives
\begin{align*}
\sum\limits_{i=1}^r [\varphi_h(S_I-s_i+o_i)-\varphi_h(S_I)]
&~\leq~ r\Delta \varphi_h(S_I) + r(2+\Delta)\alpha \varphi_h(O^\star) \\
&~\leq~ rH_r\Delta h(S_I) + r(2+\Delta)\alpha H_r h(O^\star) \tag*{(Lemma \ref{lmm:bound_g})}\\
&~\leq~ rH_r\left((2+\Delta)\alpha+\Delta\right) h(O) .  \tag*{(optimality of $O$)}
\end{align*}
Combine the Lemma \ref{lmm:local_sub}, we have 
\[h(O)-\frac{e}{e-1}h(S_I)\leq\sum\limits_{i=1}^r [\varphi_h(S_I-s_i+o_i)-\varphi_h(S_I)]\leq rH_r\left((2+\Delta)\alpha+\Delta\right) h(O).  \]
Thus
\begin{align*}
 h(S_I) &~\geq~ \left(1-\frac{1}{e}\right)\left(1-rH_r\left((2+\Delta)\alpha+\Delta\right) \right)h(O)   \\
 &~\geq ~\left(1-\frac{1}{e}\right)\left(1-r\left(\ln r+1\right)\left((2+\Delta)\alpha+\Delta\right) \right)h(O).  
\end{align*}
This completes the proof of Theorem \ref{the:main_nls}.
    
\end{proof}

\begin{corollary}{\bf{}}
\label{cor:main_nls}
Assume $n \in \Omega\left(\exp\left(\left(\frac{1}{\varepsilon}\right)^{O(1)}\right)\right)$, if $\alpha,\Delta = \frac{\varepsilon}{4r\log r}$, we have $I\leq \frac{5r\ln r}{\varepsilon}$, and with probability at least $1-(I+1)rn\delta$, the output $S_I$ of Algorithm \ref{alg:local_search} is a $\left(1-\frac{1}{e}-\varepsilon\right) $-approximate solution for problem $\max_{S\subseteq \mathcal{I}(\mathcal{M})}h(S)$,  with at most  $\left(I+1\right)rn$ calls to $\widehat{\varphi_h}$.
\end{corollary}

\section{Our algorithms and main theorems for cardinality constraints}
\label{sec:cardinality}
For different ranges of cardinalities, 
this section presents algorithms (Algorithm~\ref{alg:local_search_monotone} and~\ref{alg:local_search_monotone2}) that return a $\left(1-\frac{1}{e}-O(\varepsilon)\right)$-approximation for Problem~\ref{problem:noisy_submodular} with only $\text{Poly}\left(n, \frac{1}{\varepsilon}\right)$ queries to the noisy oracle $\Tilde{f}$.
The analyses (Theorem~\ref{thm:cardi} and~\ref{thm:cardi2}) of the algorithms constitute a proof of Theorem~\ref{thm:informal_cardi}.

\subsection{Algorithmic results for cardinality constraints when \texorpdfstring{$r\in \Omega\left(\frac{1}{\varepsilon}\right) \cap O\left(n^{\frac{1}{3}}\right)$}{Lg}}
\label{sec:cardinality1}

\noindent
We first present an algorithm (Algorithm~\ref{alg:local_search_monotone}) and its analysis (Theorem~\ref{thm:cardi}) that are applicable for all $r\in \Omega\left(\frac{1}{\varepsilon}\right) \cap O\left(n^{\frac{1}{3}}\right)$ when $n$ is sufficiently large. 

\begin{theorem}[\textbf{Algorithmic results for cardinality constraints when $r\in \Omega\left(\frac{1}{\varepsilon}\right) \cap O\left(n^{\frac{1}{3}}\right)$}]
\label{thm:cardi}
    Let $\varepsilon > 0$ and assume $n \in \Omega\left(\exp\left(\left(\frac{1}{\varepsilon}\right)^{O(1)}\right)\right)$ is sufficiently large. 
    For any $r\in \Omega(\frac{1}{\varepsilon}) \cap O\left(n^{\frac{1}{3}}\right)$, there exists an algorithm that returns a $\left(1-\frac{1}{e}-O(\varepsilon)\right)$-approximation for Problem~\ref{problem:noisy_submodular} under a $r$-cardinality constraint, with probability at least $1-O\left(\frac{1}{\log n}\right)$ and query complexity at most $O\left(r^2\log ^2r\cdot n^{\frac{3}{2}}\varepsilon^{-1}\max\{r, \log n\}\right)$ to $\tilde{f}$.
\end{theorem}
\noindent
The assumption that $n$ is sufficiently large is necessary and has also been adopted in prior works on noisy submodular optimization~\cite{hassidim2017submodular,singer2018optimization}.
We achieve a tight approximation ratio of $1-1/e-O(\varepsilon)$.
Furthermore, we only require $\text{Poly}\left(n, 1/\varepsilon\right)$ queries to $\tilde{f}$, in contrast to $\Omega\left(n^{1/\varepsilon}\right)$ for the prior greedy algorithm~\cite{singer2018optimization}.

\subsubsection{Useful notations and useful facts for Theorem~\ref{thm:cardi}}
\label{sec:notation_I}

Our algorithm and analysis are based on the following smoothing surrogate function $h$.
\begin{definition}[\bf{Smoothing surrogate function \rom{1}}]
\label{def:h_1}
For any set $S\subseteq N$, we define the smoothing surrogate function $h(S)$ as the expectation of $f(S+e)$ over a random element $e\in N$, i.e., $h(S) = \frac{1}{n}\sum_{e\in N}f(S+e)$.
\end{definition}

\noindent
The surrogate function $h(S)$ is robust to noise when $|S|$ is relatively small since at this time $h(S)$ is based on a rather large averaging set with size nearly $n$. 
Note that $h(S)$ becomes too concentrated on $f(S)$ as $|S|\approx n$, and hence, we consider this surrogate $h$ for the range $r\in O\left(n^{\frac{1}{3}}\right)$.
We now show that $h$ shares some basic properties with $f$.
\begin{lemma}[\textbf{Properties of $h$}]
\label{lem:property_of_h}
The smoothing surrogate function $h$ is monotone and submodular.
\end{lemma}
\begin{proof}
    For any set $S\subseteq N$ and $x\notin S$, by Definition~\ref{def:h_1} we have
\begin{equation}
\label{eq:h_monotone}
    h(S+x) - h(S) = \frac{1}{n}\sum_{e\in N}\left[f(S+x+e) - f(S+e)\right].
\end{equation}
This immediately implies that $h$ is monotone, since the monotonicity of $f$ implies that each term $f(S+x+e) - f(S+e)$ is non-negative.
Next, suppose that $T\subseteq S$ and $x\notin S$.
Then
\begin{align*}
    h(T+x) - h(T) &= \frac{1}{n}\sum_{e\in N}\left[f(T+x+e) - f(T+e)\right] && \text{(Eq.~\eqref{eq:h_monotone})}\\
    &\geq \frac{1}{n}\sum_{e\in N}\left[f(S+x+e) - f(S+e)\right] && \text{(submodularity of $f$)}\\
    &= h(S+x) - h(S).
    && \text{(Eq.~\eqref{eq:h_monotone})}
\end{align*}
Thus, $h$ is submodular.
\end{proof}
\noindent
Note that $h$ is implicitly constructed since we only have a value oracle to $\tilde{f}$ instead of $f$.
Hence, we construct an approximation of its auxiliary function $\varphi_h$; summarized by the following lemma.
\begin{lemma}[\textbf{Approximation of $\varphi_{h}$}]
\label{lem:hat_varphi_h}
Let $\alpha,\delta \in (0,1/2)$ and assume $n\in \Omega(\alpha^{-2}\log(\delta^{-1}))$. 
There exists a value oracle $\mathcal{O}$ to an $\left(\alpha,\delta, \mathcal{I}(r-1)\right)$-approximation $\widehat{\varphi_{h}}$ of $\varphi_{h}$, which answering $\mathcal{O}(A)$ queries at most $M = \Theta\left(\log r\cdot n^{\frac{1}{2}}\max\{r, \log n\}\right)$ times to $\tilde{f}$ for each set $A\in\mathcal{I}(r-1)$.
\end{lemma}

\noindent
The above lemma indicates that for a sufficiently large $n$, we can achieve an arbitrary accurate approximation $\widehat{\varphi_h}$ of $\varphi_h$, which queries $\tilde{f}$ only $\tilde{O}(r n^{\frac{1}{2}})$ times for each $\widehat{\varphi_h}(A)$.

More precisely, the relationship among $\alpha$, $\delta$ and $n$ can be expressed as
$\alpha^{-2}\ln(4\delta^{-1}) \leq c\cdot n\kappa^{-2}$ (See Lemma~\ref{lmm:app_diff_phi_hat1}), 
where $\kappa$ is the sub-exponential norm (Definition~\ref{def:sub_exponential}) of the noise distribution $\mathcal{D}$ depending on the parameters $c_i, \gamma_i$ of density function, and $c$ is an absolute constant.
A small $\kappa$ indicates a concentrated noise distribution. Generally speaking, the more concentrated the noise is ($\kappa \rightarrow 0$), the more accurate approximation we are able to obtain ($\alpha, \delta \rightarrow 0$).

\begin{proof}
Here we illustrate intuitively how to obtain the value oracle $\mathcal{O}$ to an $(\alpha, \delta, \mathcal{I}(r-1))$-approximation $\widehat{\varphi_h}$ of $\varphi_h$.
When $n$ is large enough, the generalized exponential tail allows the difference between $\varphi_h$ and its noisy analogue $\widetilde{\varphi_h}$ to be averaged out.
Thus we can obtain an approximation $\widehat{\varphi_h}$ of $\varphi_h$ by approximating $\widetilde{\varphi_h}$.
Since evaluating $\widetilde{\varphi_h}$ exactly will require an exponential number of value queries to $\Tilde{f}$, we use a random sampling procedure (Algorithm~\ref{alg:hat_varphi_h}) to estimate it.
With only polynomial queries to $\tilde{f}$, $\widetilde{\varphi_h}$ can be approximated to the desired accuracy.

To prove Lemma~\ref{lem:hat_varphi_h}, we first show some bounds for coefficients $\tau_A(T)$.
With these bounds, we will then prove that when $n$ is sufficiently large, $\varphi_h$ and its noisy analogue $\widetilde{\varphi_h}$ are close to each other by Bernstein's inequality for sub-exponential variables.
Finally, we use Hoeffding's inequality to show that with enough samples, $\widehat{\varphi_h}$ is close to $\widetilde{\varphi_h}$ and therefore a good approximation of $\varphi_h$.
For a set $A\subseteq N$, we rewrite the auxiliary function $\varphi_h(A)$ as a linear combination of $f(T)$ according to Definition~\ref{def:auxiliary} and~\ref{def:h_1} and obtain
\begin{equation}
    \resizebox{0.94\textwidth}{!}{$\displaystyle
    \label{eq:varphi_h_f}
        \varphi_h(A) = \frac{1}{n}\sum_{T\subseteq A}\left[(m_{|A|-1, |T|-1} + m_{|A|-1, |T|-2})|T|\cdot f(T) + \sum_{e\in N\backslash A}m_{|A|-1, |T|-1}\cdot f(T+e)\right].$}
\end{equation}
We define the averaging set of $\varphi_h(A)$ as
$$\mathcal{L}_A = \left\{T\in 2^N: T\subseteq A, \text{\ or\ } \exists\ S\subseteq A\ \text{and}\ e\in N\backslash A \text{\ s.t.\ } S+e = T\right\}.$$
%
%
For any $T\in\mathcal{L}_A$, we denote by $\tau_{A}(T)$ the linear coefficient of $f(T)$.
Then we can write $\varphi_h(A)$ as $\sum\limits_{T\in\mathcal{L}_A}\tau_{A}(T) f(T)$.

\begin{lemma}[\textbf{Useful bounds of the linear coefficients}]
\label{lem:tau_bound}
For any $A\subseteq N$ with $|A|\leq \sqrt{n}$, we have
\begin{align*}
    &(1) \sum_{T\in\mathcal{L}_{A}} \tau^2_{A}(T) \leq \frac{12}{n|A|};
    && (2) \max_{T\in\mathcal{L}_A} \tau_A(T) \leq \frac{4}{n};
    && (3) \sum_{T\in\mathcal{L}_A}\tau_A(T)\leq 2\left(\ln|A|+2\right).
\end{align*}
\end{lemma}
\begin{proof}
For notation convenience, we let $a = |A|$ and $t = |T|$ in this proof.
For any $A\subseteq N$, we have
\begin{align}
    \notag&\sum_{T\in\mathcal{L}_{A}} \tau^2_{A}(T) \\
    =&  \frac{1}{n^2}\sum^a_{t=0}{a \choose t}\left[(m_{a-1,t-1}+m_{a-1,t-2})^2\cdot t^2+m^2_{a-1,t-1}\cdot(n-a) \right] \tag*{(Eq.~\eqref{eq:varphi_h_f})}\\
    =& \frac{1}{n^2}\!\left\{{a \choose 1} \left(\int_0^1 \frac{e^p}{e-1} (1-p)^{a-1}dp\right)^2(n-a+1)\right. \tag*{(Definition~\ref{def:auxiliary})}\\
    \notag &\left.+
    \sum^a_{t=2}{a \choose t}\!\left[\left(\int_0^1 \frac{e^p}{e-1} p^{t-2}(1-p)^{a-t}dp\right)^2\!\!t^2+ \left(\int_0^1 \frac{e^p}{e-1} p^{t-1}(1-p)^{a-t}dp\right)^2\!\!(n-a)\right]\right\} \\
    \label{eq:tau2_1}\leq& \frac{e^2}{(e-1)^2n^2}\left\{\sum^a_{t=2}\left(\frac{at}{(t-1)^2}+ \frac{n-a}{at}\right)\frac{\Gamma(t)\Gamma(a-t+1)}{\Gamma(a)}  +\frac{n-a+1}{a}\right\}\\
    \notag\leq& \frac{e^2}{(e-1)^2n^2}\left\{\frac{1}{a-1}\sum^a_{t=2}\left(2a+ \frac{n-a}{a}\right)\frac{1}{{{a-2}\choose {t-2}}}  +\frac{n-a+1}{a}\right\} \tag*{($\frac{t}{t-1}\leq 2$ for $t\geq 2$)}\\
    <& \frac{6}{n^2}\left(a+\frac{n}{a}\right) \tag*{($\frac{e^2}{(e-1)^2}\leq 3$ and ${{a-2}\choose {t-2}} \geq 1$)},
\end{align}
where we use the fact that $e^p\leq e$ for $p\in[0,1]$ and a property of beta function that 
\begin{equation}
\label{eq:beta}
    B(x, y) = \int_0^1p^{x-1}(1-p)^{y-1}dp = \frac{\Gamma(x)\Gamma(y)}{\Gamma(x+y)}
\end{equation}
to obtain Ineq.~\eqref{eq:tau2_1}.
In particular, if $|A| \leq \sqrt{n}$, we have $\sum_{T\in\mathcal{L}_{A}} \tau^2_{A}(T) \leq {12}/{(n|A|)}$.

Next we show that $\max_{T\in\mathcal{L}_A} \tau_A(T) \leq {4}/{n}$.
For any set $T\subseteq A$ with $|T|\geq 2$, we have 
\begin{align*}
    \tau_A(T) &= \frac{1}{n}(m_{a-1, t-1} + m_{a-1, t-2})\cdot t  \tag*{(Eq.~\eqref{eq:varphi_h_f})}\\
    &= \frac{t}{n}\int_0^1 \frac{e^p}{e-1} p^{t-2}(1-p)^{a-t}dp \tag*{(Definition~\ref{def:auxiliary})}\\
    &\leq \frac{e}{e-1}\cdot\frac{t}{n}\cdot\frac{(t-2)!(a-t)!}{(a-1)!} \tag*{($p < 1$ and Eq.~\eqref{eq:beta})}\\
    &\leq \frac{2}{n}\cdot\frac{t}{a-1} \leq \frac{4}{n}. \tag*{($\frac{e}{e-1} \leq 2$, ${{a-2}\choose{t-2}}\geq 1$ and $\frac{t}{a-1} \leq 2$)}
\end{align*}
Similarly, for a set $T$ that only contains a single element in $A$,
$$\tau_A(T) = \frac{m_{a-1, 0}}{n} = \frac{1}{n}\int_0^1 \frac{e^p}{e-1} (1-p)^{a-1}dp \leq \frac{2}{an}.$$
For any $T\in\mathcal{L}_A\setminus 2^A$, we have
\begin{align*}
    \tau_A(T) &= \frac{1}{n}\cdot m_{a-1,t-2} && \text{(Eq.~\eqref{eq:varphi_h_f})}\\
    &= \frac{1}{n}\int_0^1 \frac{e^p}{e-1} p^{t-2}(1-p)^{a-t+1}dp && \text{(Definition~\ref{def:auxiliary})}\\
    &\leq \frac{e}{(e-1)n}\cdot\frac{1}{{{a}\choose {t-2}}}\cdot\frac{1}{a-t+2}\leq \frac{2}{n}. && \text{(Eq.~\eqref{eq:beta} and $2\leq t\leq a+1$)}
\end{align*}
Combining these three parts gives $\max_{T\in\mathcal{L}_A} \tau_A(T) \leq {4}/{n}$.
Finally we can bound $\sum_{T\in\mathcal{L}_A}\tau_A(T)$ as below:
\begin{align*}
    &\sum_{T\in\mathcal{L}_A}\tau_A(T)\\
    =& \frac{1}{n}\sum^a_{t=0}{a \choose t}\left[(m_{a-1, t-1} + m_{a-1, t-2})t + m_{a-1, t-1}(n-a)\right]\tag*{(Eq.~\eqref{eq:varphi_h_f})}\\
    =&\frac{1}{n}\left\{\sum^a_{t=2}{a \choose t}\left[t\int_0^1 \frac{e^p}{e-1} p^{t-2}(1-p)^{a-t}dp+(n-a)\int_0^1 \frac{e^p}{e-1} p^{t-1}(1-p)^{a-t}dp\right]\right.\\
    & \left.+{a \choose 1} (n-a+1)\int_0^1 \frac{e^p}{e-1} (1-p)^{a-1}dp\right\} \tag*{(Definition~\ref{def:auxiliary})}\\
    \leq& \frac{e}{(e-1)n}\cdot\left[\sum^a_{t=2}\left(\frac{a}{t-1}+\frac{n-a}{t}\right)+(n-a+1)\right] && \tag*{($p < 1$)}\\
    \leq& 2(\ln a+2). && \tag*{($\frac{e}{e-1} \leq 2$ and $\sum^a_{t=2}\frac{1}{t-1} \leq \ln a + 1$)}
\end{align*}
\end{proof}
\noindent
We now use concentration results of sub-exponential distributions to bound the difference between $\varphi_h$ and its noisy analogue $\widetilde{\varphi_h}$.
\begin{definition}[\textbf{Sub-exponential distribution}]
\label{def:sub_exponential}
The sub-exponential norm of $X \in  \mathbb{R}$ is 
\[\left\|X\right\|_{\psi_1} = \inf \{t > 0 : \mathbb{E} \exp(|X|/t) \leq 2\} .\]
If $\left\|X\right\|_{\psi_1}$
is finite, we say that $X$ is sub-exponential.
\end{definition}
\begin{claim}
\label{clm:subexp}
Let $\xi$ be a random variable with a generalized exponential tail distribution $\mathcal{D}$.
Then $\xi-1$ is sub-exponential.
\end{claim}
\begin{proof}
If the support of $\mathcal{D}$ is bounded by some constant $b$, then let $\kappa = (b+1)/\ln 2$, and we have $\left\|\xi-1\right\|_{\psi_1}\leq \kappa$.
For the case where $\mathcal{D}$ does not have bounded support, recall that the probability density function $\rho(x)$ of $\mathcal{D}$ is $\exp(-\sum_ic_ix^{\gamma_i})$. 
When 
$$x\geq \max\left\{1, \left(\frac{2\sum_{i=1}|c_i|}{c_0}\right)^{\frac{1}{\gamma_0-\gamma_1}}\right\},$$
the term $\frac{1}{2}c_0x^{\gamma_0}$ dominates the rest of the terms, and thus $\rho(x) \leq e^{-\frac{1}{2}c_0x^{\gamma_0}} \leq e^{-\frac{1}{2}c_0x}$.
Let 
$$\kappa = \max\left\{2,\  \frac{2}{c_0}\cdot\ln\left(\frac{8}{c_0}\right)+1,\  \ln\left(\frac{2}{3}\right)\cdot\left(\frac{2\sum_{i=1}|c_i|}{c_0}\right)^{\frac{1}{\gamma_0-\gamma_1}}  \right\}.$$
It is straightforward to verify that $\mathbb{E}\left[\exp\left(\frac{|\xi-1|}{\kappa}\right)\right] \leq 2$.
%
%
Hence $\xi-1$ is a sub-exponential variable with a constant sub-exponential norm $\left\|\xi-1\right\|_{\psi_1} \leq \kappa$.
\end{proof}
\noindent
For sub-exponential random variables, we have the following Bernstein's inequality.
\begin{lemma}[\textbf{Bernstein's inequality for sub-exponential variables~\cite{roman2018high}}]
\label{lem:bern_exp}
Let $X_1,\dots, X_m$ be independent, mean zero, sub-exponential random variables, and $a=(a_1,\dots,a_m)\in \mathbb{R}^m$.
Then, for every $t>0$, we have
\begin{equation}
\label{eq:Bernstein}
    \mathbb{P}\left[\left|\sum^m_{i=1}a_iX_i\right|\geq t\right] \leq 2\exp\left[-c\min\left(\frac{t^2}{K^2\left\|a\right\|^2_2}, \frac{t}{K\left\|a\right\|_{\infty}}\right)\right],
\end{equation}
where $K = \max_i\left\|X_i\right\|_{\psi_1}$ and $c$ is a constant.
\end{lemma}
\noindent
With Bernstein's inequality, we can bound the difference between $\varphi_h$ and $\widetilde{\varphi_h}$ by the following lemma.
\begin{lemma}[\textbf{Difference between $\varphi_h$ and $\widetilde{\varphi_h}$}]
\label{lmm:app_diff_phi_hat1}
Suppose that $n \geq 192\kappa^2c^{-1}\cdot \alpha^{-2}\ln(4\delta^{-1})$, where $\kappa$ denotes the sub-exponential norm of the noise multiplier and $c$ is the constant in Ineq.~\eqref{eq:Bernstein}.
For any set $A\in \mathcal{I}(r-1)$, we have $$\mathbb{P}\left[|\widetilde{\varphi_h}(A)-\varphi_h(A)| > \frac{\alpha}{2}\cdot\max_{S\in \mathcal{I}(r-1)}\varphi_h(S)\right] \leq \frac{\delta}{2}.$$
\end{lemma}
\begin{proof}
Let us fix a set $A\in\mathcal{I}(r-1)$.
From monotonicity and submodularity of $f$, for any $T\in\mathcal{L}_A$,
\begin{equation}
\label{eq:max}
    f(T) \leq f(A) + \max_{e\in N}f(e) \leq 2\cdot\max_{S\in \mathcal{I}(r-1)}f(S) \leq 2\cdot\max_{S\in \mathcal{I}(r-1)}\varphi_h(S).
\end{equation}
%
%
Let $X_T = \xi(T)-1$ and $a_T = \tau_{A}(T)f(T)$.
By Lemma~\ref{lem:tau_bound}, we have
\begin{align*}
\sum_{T\in\mathcal{L}_A}\left(\tau_{A}(T)f(T)\right)^2
\leq 4\left(\sum_{T\in\mathcal{E}_{A}}\tau^2_{A}(T)\right)\left(\max_{S\in \mathcal{I}(r-1)}\varphi_h(S)\right)^2\leq \frac{48}{|A|n}\cdot\left(\max_{S\in \mathcal{I}(r-1)}\varphi_h(S)\right)^2
\end{align*}
and
$$\max_{T\in\mathcal{L}_A}\tau_{A}(T)f(T) \leq 2\max_{T\in\mathcal{L}_A}\tau_{A}(T) \cdot \max_{S\in \mathcal{I}(r-1)}\varphi_h(S)\leq \frac{8}{n}\cdot\max_{S\in \mathcal{I}(r-1)}\varphi_h(S).$$
Since we have 
$\alpha|A| < \varepsilon/(4H_{r-1}) < 12\kappa$,
applying Lemma~\ref{lem:bern_exp} gives
$$\mathbb{P}\left[|\widetilde{\varphi_h}(A)-\varphi_h(A)|\geq 
\frac{\alpha}{2}\cdot\max_{S\in \mathcal{I}(r-1)}\varphi_h(S)\right] 
\leq 2\exp\left(-c\cdot\frac{n\alpha^2|A|}{192\kappa^2}\right).
$$
By assumption that
$n\geq 192\kappa^2c^{-1}\cdot \varepsilon^{-2}_0\ln(4\delta^{-1})$,
the probability above is at most $\delta/2$.
\end{proof}
\noindent
The random sampling procedure that we use to estimate $\widetilde{\varphi_h}$ is presented as Algorithm~\ref{alg:hat_varphi_h}.

\begin{algorithm}[htp] 
	\caption{Approximation of $\varphi_h$} 
	\label{alg:hat_varphi_h}
	\DontPrintSemicolon
    \SetNoFillComment
    \SetKwInOut{Input}{Input}
		\Input{Noisy oracle $\Tilde{f}$, a set $A\in\mathcal{I}(r-1)$, number of samples $M$.}
         Construct a distribution $\nu(T)$ on set $\mathcal{L}_A$ such that
		$\nu(T) = {\tau_A(T)}/{\left(\sum_{T'\in \mathcal{L}_A} \tau_{A}(T')\right)}$\;
		 Sample $M$ sets $T_1,\dots,T_{M}$ in $\mathcal{L}_A$ according to the distribution $\nu(T)$\;
		 $\widehat{\varphi_h}(A)\leftarrow\left(\sum_{T'\in \mathcal{L}_A} \tau_A(T')\right)\cdot \frac{1}{M}\sum_{j=1}^M\Tilde{f} (T_j)$\;
	\Return $\widehat{\varphi_h}(A)$
\end{algorithm}
We plan to use Hoeffding's inequality to show that $\widehat{\varphi_h}(A)$ and $\widetilde{\varphi_h}(A)$ are very close to each other, and thus $\widehat{\varphi_h}$ is a $(\alpha, \delta, \mathcal{I}(r-1))$-approximation of $\varphi_h$.
\begin{lemma}[\textbf{Hoeffding’s inequality}]
\label{lmm:Hoeffding}
Let $X_1,\cdots , X_N$ be independent bounded
random variables such that $X_i \in [a_i, b_i]$, where $-\infty < a_i \leq b_i < \infty$. 
Then
\[\mathbb{P}\left[\left|\frac{1}{N}\sum\limits_{i=1}^N(X_i-\mathbb{E}X_i)\right|\geq t\right]\leq 2\exp \left[-\frac{2N^2t^2}{\sum^N_{i=1}(b_i-a_i)^2}\right].\]
\end{lemma}
\noindent To use Hoeffding's inequality, we need a bound $b$ for noise multipliers.
If the noise distribution $\mathcal{D}$ has bounded support, this bound is natural.
When the support of $\mathcal{D}$ is not bounded, we let 
$$b = \frac{2}{c_0}\ln\left(c_0^{-1}2^{|A|+3}n\delta^{-1}\right).$$
%
%
%
%
When $n$ is sufficiently large,
$$\mathbb{P}\left[\xi\geq b\right] \leq \int_{b}^{\infty} \exp\left(-\frac{1}{2}c_0x\right)dx \leq \frac{\delta}{2^{|A|+2}n}.$$
The probability that all $\xi(T)\ (T\in\mathcal{L}_A)$ are bounded by $b$ is
$$\left(1-\frac{\delta}{2^{r} n}\right)^{|\mathcal{L}_A|} \geq 1- |\mathcal{L}_A|\cdot\frac{\delta}{2^{|A|+2} n} \geq 1- \frac{\delta}{4}.$$


\begin{lemma}[\textbf{Difference between $\widetilde{\varphi_h}$ and $\widehat{\varphi_h}$}]
\label{lmm:app_diff_tilde_hat_1}
For any set $A\in \mathcal{I}(r-1)$, we assume $M \geq c^{\frac{1}{2}}\kappa^{-1}\cdot b(\ln(r-1)+2)\sqrt{n}$, where $c$ is the constant in Ineq.~\eqref{eq:Bernstein}, and suppose that all $\xi(T)\ (T\in\mathcal{L}_A)$ are bounded by $b$. 
Then
$$\mathbb{P}\left[|\widehat{\varphi_h}(A)-\widetilde{\varphi_h}(A)| > \frac{\alpha}{2}\cdot\max_{S\in \mathcal{I}(r-1)}\varphi_h(S)\right] \leq \frac{\delta}{4}.$$
\end{lemma}
%
\begin{proof}
We apply Hoeffding's inequality by letting $X_j = \left(\sum_{T'\in\mathcal{L}_A}\tau_A(T')\right)\Tilde{f}(T_j)$ and obtain
$$
\mathbb{P}\left[|\widehat{\varphi_h}(A)- \widetilde{\varphi_h}(A)|\geq \frac{\alpha}{2}\cdot\max_{S\in \mathcal{I}(r-1)}\varphi_h(S)\right] 
    \leq 2\exp\left[-\frac{M^2\varepsilon^2_0}{8b^2\left(\sum_{T\in\mathcal{E}_{A}}\tau_{A}(T)\right)^2}\right].
$$
By assumption that $M \geq c^{\frac{1}{2}}\kappa^{-1}\cdot b(\ln(r-1)+2)\sqrt{n}$ and $n \geq 192\kappa^2c^{-1}\cdot \alpha^{-2}\ln(4\delta^{-1})$, the probability above is at most ${\delta}/{4}$.
\end{proof}
\noindent
For any $A\in\mathcal{I}(r-1)$, by taking a union bound, $|\widehat{\varphi_h}(A)-\varphi_h(A)|\leq \alpha\max_{S\in \mathcal{I}(r-1)}\varphi_h(S)$ fails with probability at most ${\delta}/{2}+{\delta}/{4}+(1-{\delta}/{4}){\delta}/{4} \leq \delta$.
The lemma is proved.
\end{proof}

\subsubsection{Algorithm for Theorem~\ref{thm:cardi}}

\begin{algorithm}[htp]
 \caption{Noisy local search under a small cardinality constraint}
  \label{alg:local_search_monotone}
\DontPrintSemicolon
\SetNoFillComment
  \SetKwInOut{Input}{Input}
  \Input{a value oracle to $\Tilde{f}$, budget $r\in \Omega\left(\frac{1}{\varepsilon}\right) \cap O\left(n^{\frac{1}{3}}\right)$ and $\varepsilon\in (0,1/2)$}
  %
        Let $\widehat{\varphi_h}$ be a $\left(\alpha=\frac{\varepsilon}{4r\ln r}, \delta=\frac{1}{(I+1)(r-1)n^2}, \mathcal{I}(r-1)\right)$-approximation of $\varphi_h$ as in Lemma~\ref{lem:hat_varphi_h}\; 
    	$S_L\leftarrow\mathtt{NLS}\left(\widehat{\varphi_h}, \mathcal{I}(r-1), \Delta = \frac{\varepsilon}{4r\ln r}\right)$ \Comment*[r]{Local search phase}
		$S_M \leftarrow S_L+ \arg\max_{e\in N\backslash S_L} \Tilde{f}(S_L+e)$\Comment*[r]{$\Tilde{f}$-maximization phase}
\Return $S_{M}$
\end{algorithm}
We present Algorithm~\ref{alg:local_search_monotone} that contains two phases: a local search phase (Line 2) and a $\Tilde{f}$-maximization phase (Line 3).
We first run a non-oblivious local search procedure (Algorithm~\ref{alg:local_search}) under the $(r-1)$-cardinality constraint (Line 2).
We set $\delta = \frac{1}{(I+1)(r-1)n^2}$ by  Corollary~\ref{cor:main_nls} when applying Algorithm~\ref{alg:local_search} to obtain a $\left(1-\frac{1}{e}-\varepsilon\right)$-approximation solution $S_L$ for $h$-maximization.
At the $\Tilde{f}$-maximization phase (Line 3), we obtain $S_M$ by selecting an additional element $e\in N\setminus S_L$ that maximizes $\tilde{f}(S_L+e)$.
This guarantees that $f(S_M) \geq (1-\varepsilon)h(S_L)$ with high probability (Lemma~\ref{lem:add}), which results in the approximation ratio $1-\frac{1}{e}-O(\varepsilon)$ in Theorem~\ref{thm:cardi}.
Given the smoothing surrogate function $h(S)$ (Definition 4.2), a natural idea is to simply apply local search (Algorithm 1) to optimize it under the $r$-cardinality constraint. However, this idea does not yield a provable approximation of $\max_{S:|S|\leq r} f(S)$, since it is not easy to control the contribution of the additional element $e\in N$ to $h(S)$ and it is possible that $h(S)\gg f(S)$.
To handle this difficulty, Algorithm~\ref{alg:local_search_monotone} first runs a local search procedure under the $(r-1)$-cardinality constraint, and then selects an additional element with a "large enough" margin at the $\tilde{f}$-maximization phase.
This idea enables us to control the loss induced by the surrogate $h$  within $\frac{1}{r}\cdot \max_{S:|S|\leq r} f(S)$.

\subsubsection{Proof of Theorem~\ref{thm:cardi}}
We first analyze the query complexity and then prove the approximation performance. 

\paragraph{Query complexity of Algorithm~\ref{alg:local_search_monotone}.}
In Line 3, we make $(n-r+1)$ calls to the noisy oracle $\Tilde{f}$ in total.
By Corollary~\ref{cor:main_nls}, the total number of queries to $\widehat{\varphi_h}$ in Line 2 is at most $(I+1)(r-1)n$ and $I\leq \frac{5r\ln r}{\varepsilon}$.
Combining with Lemma~\ref{lem:hat_varphi_h}, the query complexity to $\tilde{f}$ is upper bounded by
$(n-r+1) + M(I+1)(r-1)n = O\left(r^2\log ^2r\cdot n^{\frac{3}{2}}\varepsilon^{-1}\max\{r, \log n\}\right)$.
This matches Theorem~\ref{thm:cardi}.

\paragraph{Approximation ratio and success probability of Algorithm~\ref{alg:local_search_monotone}.}
Let $O_h \in \arg\max_{S\in \mathcal{I}(r-1)} h(S)$ represent the $(r-1)$-set whose value of $h$ is the largest.
Following from Theorem~\ref{the:main_nls}, we have
$h(S_L) \geq \left(1-\frac{1}{e}-\varepsilon\right)h(O_{h})$ with probability at least $1-\frac{1}{n}$.
We denote by $O_f\in \arg\max_{S\in \mathcal{I}(r)}f(S)$ the optimal solution to $f$.
By the submodularity of $f$, $O_f$ has a subset $\Tilde{O}_f$ with $r-1$ elements such that $f(\Tilde{O}_f) \geq \left(1-\frac{1}{r}\right)f(O_f)$.
Then we have
\begin{equation}
\label{eq:S_L_2}
    h(S_L) \geq \left(1-\frac{1}{e}-\varepsilon\right)h(\Tilde{O}_{f}) \geq \left(1-\frac{1}{e}-\varepsilon\right)f(\Tilde{O}_{f}) \geq \left(1-\frac{1}{r}\right)\left(1-\frac{1}{e}-\varepsilon\right)f(O_f),
\end{equation}
where the first inequality follows from $h(S_L) \geq \left(1-\frac{1}{e}-\varepsilon\right)h(O_{h})$ and the second from monotonicity of $f$.
Since $r$ is assumed to be $\Omega\left(\frac{1}{\varepsilon}\right)$, $h(S_L)$ is a $\left(1-\frac{1}{e}-O(\varepsilon)\right)$-approximation of $f(O_f)$.

Recall that $h(S_L)$ is the expectation of $f(S_L+e)$ over a random element $e\in N$.
Ineq.~\eqref{eq:S_L_2} already proves a claim that uniformly randomly adding an element $e\in N$ to $S_L$ achieves an approximation ratio of $1-\frac{1}{e}-O(\varepsilon)$ in expectation for maximizing $f$, i.e.,
$\mathbb{E}[f(S_L + e)]\geq \left(1-\frac{1}{e}-O(\varepsilon)\right)f(O_f).$
Moreover, we can convert this claim to a with-high-probability claim by the following lemma.

\begin{lemma}[\textbf{Approximation at the $\Tilde{f}$-maximization phase}]
\label{lem:add}
With probability $1-O\left(\frac{
1}{\log n}\right)$, we have 
$f(S_M) \geq (1-\varepsilon)h(S_L).$
\end{lemma}
\noindent
This lemma indicates that a \textit{bad} element $e\in N\setminus S_L$ with $f(S_L+e) < (1-\varepsilon)h(S_L)$ is unlikely to be chosen at the $\Tilde{f}$-maximization phase.
This is because the local search procedure guarantees that the number of \textit{good} elements with $f(S_L+e)\geq (1-\varepsilon)h(S_L)$ is almost of the same order as the bad ones. 

\begin{proof}
    We denote by $e^{\star}\in \arg\max_{e\in N\backslash S_L}\Tilde{f}(S_L+e)$ the element added to $S_L$ at the $\Tilde{f}$-maximization phase by Algorithm~\ref{alg:local_search_monotone}.
We define two kinds of elements in $N\backslash S_L$.
Say that an element $e$ is \textit{good} if
\begin{equation}
\label{eq:good_elements}
    f(S_L+e) \geq (1-\frac{\varepsilon}{2}) h(S_L),
\end{equation}
and $e$ is \textit{bad} if 
$$f(S_L+e) < (1-\varepsilon)h(S_L).$$
The set of good elements and that of bad ones are denoted as $\mathcal{G}$ and $\mathcal{B}$, respectively.
Our goal is to prove that $e^{\star}$ is not bad with probability $1-O\left(\frac{1}{\log n}\right)$.
\noindent
First we use the following lemma to quantify the number of good elements.
\begin{lemma}[\textbf{The number of good elements}]
\label{lem:no_of_good}
    With probability $1-\frac{1}{n}$, there are at least $\frac{\varepsilon n}{16H_{r-1}}$ good elements in $N\backslash S_L$, where $H_{r-1} = \sum^{r-1}_{i=1} \frac{1}{i}$.
\end{lemma}
\begin{proof}
Suppose by contradiction that $|\mathcal{G}| < \frac{\varepsilon n}{16 H_{r-1}}$.
We will show that if there are only a few good elements in $N\backslash S_L$, a good element with very large $f$-value must exist.
However, this contradicts how Algorithm~\ref{alg:local_search_monotone} does local search, which should have changed some element in $S_L$ for the one with such large $f$-value.

By construction of function $h$ and definition of good elements, we have
\begin{align*}
    h(S_L) 
    &= \frac{1}{n}\left(\sum_{e\in \mathcal{G}}f(S_L+e) + \sum_{e\in N\backslash\mathcal{G}} f(S_L+e) \right) && \text{(Definition~\ref{def:h_1})}\\ 
    &< \frac{1}{n}\left(\sum_{e\in  \mathcal{G}}f(S_L+e) + |N\backslash\mathcal{G}|\left(1-\frac{\varepsilon}{2}\right)h(S_L) \right). && \text{(Ineq.~\eqref{eq:good_elements})}
\end{align*}
Rewriting the inequality above gives
$$\sum_{e\in\mathcal{G}}f(S_L+e) > \left(n-\left(1-\frac{\varepsilon}{2}\right)|N\backslash\mathcal{G}|\right)\cdot h(S_L) > \frac{\varepsilon n}{2}\cdot h(S_L),$$
where the last inequality holds since $|N\backslash\mathcal{G}|<n$.
Hence, there must exist a good element $e_0$ such that 
\begin{equation}
\label{eq:f_max}
    f(S_L+e_0) \geq \frac{\varepsilon n}{2|\mathcal{G}|}\cdot h(S_L) > 8H_{r-1}\cdot h(S_L).
\end{equation}
%
%
%
%

Consider the set $S_L+e_0-e'_0$, where $e'_0$ is an arbitrary element in $S_L$.
On the one hand, we must have
\begin{equation}
\label{eq:local_step}
    \widehat{\varphi_h}(S_L+e_0-e'_0) \leq (1+\alpha)\widehat{\varphi_h}(S_L).
\end{equation}
Otherwise, Algorithm~\ref{alg:local_search_monotone} would have changed $e'_0$ for $e_0$ rather than returning $S_L$.
We translate Ineq.~\eqref{eq:local_step} to that of $\varphi_h$.
With probability $1-O(\frac{1}{n})$, 
\begin{align*}
    \varphi_h(S_L+e_0-e'_0)
    &\leq  \widehat{\varphi_h}(S_L+e_0-e'_0) + \notag\alpha\max_{S\in \mathcal{I}(r-1)}\varphi_h(S) && \text{(Lemma~\ref{lem:hat_varphi_h})}\\
    \notag &\leq  (1+\alpha)\widehat{\varphi_h}(S_L) + \alpha\max_{S\in \mathcal{I}(r-1)}\varphi_h(S) && \text{(Ineq.~\eqref{eq:local_step})} \\
    \notag &\leq  (1+\alpha)\left[\varphi_h(S_L)+2\alpha\max_{S\in \mathcal{I}(r-1)}\varphi_h(S)\right] && \text{(Lemma~\ref{lem:hat_varphi_h})}\\
    &\leq (1+\alpha)\left(1+\frac{4\alpha}{1-2r\alpha}\right)\varphi_h(S_L) && 
    \text{(Lemma~\ref{lmm:init_solution})}\\
    &\leq  (1+10\alpha)\varphi_h(S_L). && \text{($r\alpha < {1}/{4}$)}
\end{align*}
On the other hand, for any $\alpha < \frac{1}{4}$ and $r \geq 2$,
\begin{align*}
    \varphi_h(S_L+e_0-e'_0) 
    &\geq h(S_L+e_0-e'_0) &&\text{(Lemma~\ref{lmm:bound_g})}\\
    &\geq f(S_L+e_0) - f(e'_0) && \text{(monotonicity and submodularity of $f$)}\\
    &\geq (8H_{r-1}-1)h(S_L) && \text{(Ineq.~\eqref{eq:f_max}) and $f(e'_0) \leq f(S_L) \leq h(S_L)$)}\\
    &>(1+10\alpha)\frac{e}{e-1}H_{r-1}h(S_L)&& \text{($\alpha<1/4$ and $\frac{e}{e-1}<2$)}\\
    &\geq (1+10\alpha)\varphi_h(S_L). &&\text{(Lemma~\ref{lmm:bound_g})}
\end{align*}
This constitutes a contradiction, and the lemma is proved.
\end{proof}
\noindent
Let $\xi^*_{\mathcal{G}} = \max_{e\in \mathcal{G}} \xi(S_L+e)$ and $\xi^*_{\mathcal{B}} = \max_{e\in \mathcal{B}} \xi(S_L+e)$ denote the largest noise multiplier on the sets belonging to $\{S_L+e: e\in \mathcal{G}\}$ and $\{S_L+e: e\in \mathcal{B}\}$, respectively. 
Next we show that with sufficient good elements guaranteed by Lemma~\ref{lem:no_of_good}, $\xi^*_{\mathcal{G}}$ and $\xi^*_{\mathcal{B}}$ are close to each other with high probability.

If the distribution has bounded support and there is an atom at the supremum with some probability $p$, it is clear that both $\xi^*_{\mathcal{G}}$ and $\xi^*_{\mathcal{B}}$ reaches the supremum with high probability if $n$ is sufficiently large.
Hence we focus on the case where the distribution $\mathcal{D}$ does not have bounded support in the following analysis.
Recall that the probability density function $\rho(x)$ of the noise distribution $\mathcal{D}$ is $\exp(-\sum_{i}c_ix^{\gamma_i})$. 
We define two threshold $m_{\mathcal{G}}$ and $M_{\mathcal{B}}$:
$$m_{\mathcal{G}} = \left[\frac{1}{(1+\beta)c_0}\cdot\ln \left(\frac{|\mathcal{G}|}{\gamma_0\ln n}\right)\right]^{\frac{1}{\gamma_0}}$$
and
$$M_{\mathcal{B}} = \left[\frac{1}{(1-\beta)c_0}\ln \left(\frac{|\mathcal{B}|\ln n}{\gamma_0}\right)\right]^{\frac{1}{\gamma_0}},$$
where we set $\beta = \left(\frac{1}{\ln n}\right)^{\frac{\gamma_0-\gamma_1}{2 \gamma_0}}$.
Note that $m_{\mathcal{G}}$ and $M_{\mathcal{B}}$ are very close to each other in a sense that
$$\frac{m_{\mathcal{G}}}{M_{\mathcal{B}}} \geq \left[\frac{1-\beta}{1+\beta}\cdot\frac{\ln n -\ln(1/\varepsilon)-\ln\ln n -\ln(16\gamma_0H_{r-1})}{\ln n + \ln\ln n -\ln \gamma_0}\right]^{\frac{1}{\gamma_0}} > 1-\frac{\varepsilon}{2},$$
holds when $n\geq 6\exp\left(\left(\frac{1}{\varepsilon}\right)^\frac{2\gamma_0}{\gamma_0-\gamma_1}\right)$.
Our goal is to show that with high probability, $m_{\mathcal{G}}$ is a lower bound for $\xi^*_{\mathcal{G}}$, and $M_{\mathcal{B}}$ is an upper bound for $\xi^*_{\mathcal{B}}$.
Before that, we need the following upper and lower bounds for the probability density function $\rho(x)$ of a noise distribution. 
\begin{lemma}[\textbf{Upper and lower bounds for $\rho(x)$~\cite{singer2018optimization}}]
\label{lem:noise_bound}
For any noise distribution $\mathcal{D}$ has a generalized exponential tail, there exists $n_0$ such that for any $n>n_0$ and $x \geq \left(\frac{2\sum_{i=1}|c_i|}{c_0}\right)^{\frac{1}{\gamma_0-\gamma_1}}\cdot(\ln n)^{\frac{1}{2\gamma_0}}$, the probability density function $\rho(x) = \exp(-\sum_ic_ix^{\gamma_i})$ of $\mathcal{D}$ has the following upper and lower bound:
$$(1+\beta)c_0x^{\gamma_0-1}e^{-(1+\beta)c_0x^{\gamma_0}} \leq \rho(x) \leq (1-\beta)c_0x^{\gamma_0-1}e^{-(1-\beta)c_0x^{\gamma_0}}.$$
\end{lemma}
\begin{lemma}[\textbf{Bounds for $\xi^*_{\mathcal{G}}$ and $\xi^*_{\mathcal{B}}$}]
\label{lem:multiplier_bound}
With probability $1-O(\frac{1}{\log n})$, $\xi^*_{\mathcal{G}} \geq m_{\mathcal{G}}$ and $\xi^*_{\mathcal{B}} \leq M_{\mathcal{B}}$.
\end{lemma}
\begin{proof}
For a single sample $\xi$ drawn from noise distribution $\mathcal{D}$, we have
$$\mathbb{P}\left[\xi \leq m_{\mathcal{G}}\right] = 1-\int^{\infty}_{m_{\mathcal{G}}}\rho(x)dx \leq 1-\int^{\infty}_{m_{\mathcal{G}}}(1+\beta)c_0x^{\gamma_0-1}e^{-(1+\beta)c_0x^{\gamma_0}}dx= 1-\frac{\ln n}{|\mathcal{G}|},$$
which follows from the lower bound shown in Lemma~\ref{lem:noise_bound} when $n$ is sufficiently large.
The probability that all $\xi(S_L+e)\ (e\in \mathcal{G})$ are bounded by $m_{\mathcal{G}}$ is at most
$$\left(1-\frac{\ln n}{|\mathcal{G}|}\right)^{|\mathcal{G}|}\leq \frac{1}{e^{\ln n}} = \frac{1}{n}.$$
Hence $\xi^*_{\mathcal{G}} \geq m_{\mathcal{G}}$ holds with probability at least $1-\frac{1}{n}$.

The probability that a single noise multiplier $\xi$ is bounded by $M_{\mathcal{B}}$ is at least
$$\mathbb{P}\left[\xi \leq M_{\mathcal{B}}\right] = 1-\int^{\infty}_{M_{\mathcal{B}}}\rho(x)dx \geq 1-\int^{\infty}_{M_{\mathcal{B}}}(1-\beta)c_0x^{\gamma_0-1}e^{-(1-\beta)c_0x^{\gamma_0}}dx= 1-\frac{1}{|\mathcal{B}|\ln n},$$
which follows from the upper bound shown in Lemma~\ref{lem:noise_bound} when $n$ is sufficiently large.
Then we obtain
$$\mathbb{P}[\xi^*_{\mathcal{B}}\leq M_{\mathcal{B}}] \geq \left(1-\frac{1}{|\mathcal{B}|\ln n}\right)^{|\mathcal{B}|}\geq 1-\frac{1}{\ln n}.$$

Therefore, by a union bound, $\xi^*_{\mathcal{G}} \geq m_{\mathcal{G}}$ and $\xi^*_{\mathcal{B}} \leq M_{\mathcal{B}}$ hold with probability $1-O(\frac{1}{\log n})$.
\end{proof}
\noindent
With Lemma~\ref{lem:multiplier_bound}, we have
$$\max_{e\in \mathcal{G}}\Tilde{f}(S_L+e) \geq m_{\mathcal{G}}h(S_L) > \left(1-\frac{\varepsilon}{2}\right)M_{\mathcal{B}}h(S_L) > \max_{e\in \mathcal{B}}\Tilde{f}(S_L+e)$$
with probability $1-O(\frac{1}{\log n})$.
Thus a bad element will not be selected by Algorithm~\ref{alg:local_search_monotone} at the $\Tilde{f}$-maximization phase, i.e., $f(S_M) \geq (1-\varepsilon)h(S_L)$.
\end{proof}
\noindent
By Ineq.~\eqref{eq:S_L_2} and Lemma~\ref{lem:add}, we complete the proof of Theorem \ref{thm:cardi}.

\subsection{Algorithmic results for cardinality constraints when \texorpdfstring{$r\in \Omega\left(n^{\frac{1}{3}}\right)$}{Lg}}
\label{sec:cardinality2}

We present in this subsection an algorithm (Algorithm~\ref{alg:local_search_monotone2}) and its analysis (Theorem~\ref{thm:cardi2}) that can be applied to $r\in \Omega\left(n^{\frac{1}{3}}\right)$ when $n$ is sufficiently large.

\begin{theorem}[\textbf{Algorithmic results for cardinality constraints when  $r\in  \Omega\left(n^{\frac{1}{3}}\right) $}]
\label{thm:cardi2}
    Let $\varepsilon>0$ and assume $n \in \Omega\left(\frac{1}{\varepsilon^4}\right)$ is sufficiently large. 
    For any $r\in  \Omega\left(n^{\frac{1}{3}}\right)$, there exists an algorithm that returns a $\left(1-\frac{1}{e}-O(\varepsilon)\right)$-approximation for  Problem~\ref{problem:noisy_submodular} under a $r$-cardinality constraint, with probability $1-O\left(\frac{1}{n^2}\right)$ and query complexity at most $O(n^6\varepsilon^{-1})$ to $\Tilde{f}$.
\end{theorem}
\noindent
The above theorem addresses the remaining range of $r$ in Theorem~\ref{thm:cardi}, and also achieves a tight approximation ratio of $1-\frac{1}{e}-O(\varepsilon)$ with $\text{Poly}(n,1/\varepsilon)$ queries to $\tilde{f}$.

\subsubsection{Useful notations and useful facts for Theorem~\ref{thm:cardi2}}
\label{sec:notation_II}
Our algorithm and analysis in this subsection are based on another smoothing surrogate function $h$.
\begin{definition}[\bf{Smoothing surrogate function \rom{2}}]
\label{def:h_2}
Given a subset $H\subseteq N$, for any set $S\subseteq N\setminus H$, we define the smoothing surrogate function $h_H(S)$ as 
   $h_H(S)= \sum_{H_j\subseteq H} f(S\cup H_j)/2^{|H|}$.
\end{definition}
\noindent
Intuitively, if $|H|$ is sufficiently large, the surrogate function $h$ is robust to noise as it is based on a large averaging set with size $2^{|H|}$. 
Throughout the rest of this subsection, we consider the case that $|H|\in \Omega(\log n)$.
Similar to Section~\ref{sec:notation_I}, we give some basic properties of $h_H$.

\begin{lemma}[\textbf{Properties of $h_H$}]
\label{lmm:property_of_h2}
For any $H\subseteq N$, the smoothing surrogate function $h_H$ is monotone and submodular, and for all $S\subseteq N\setminus H$, $h_H(S)\geq \frac{1}{2} f(S\cup H) + \frac{1}{2}f(S)$.
\end{lemma}
\begin{proof}
    For any set $S\subseteq N$ and $x\not\in S$, by Definition~\ref{def:h_2}, we have 
\[h(S+x)-h(S) =\frac{1}{2^{|H|}} \sum\limits_{H_j\subseteq H} \left[f(S\cup H_j+x) - f(S\cup H_j)\right], \]
which immediately indicates that $h$ is monotone.
Moreover, for a set $T\subseteq S$, 
\begin{align*}
    h(T+x)-h(T) &=\frac{1}{2^{|H|}} \sum\limits_{H_j\subseteq H} \left[f(T\cup H_j+x)- f( T\cup H_j)\right]\\
    &\geq \frac{1}{2^{|H|}}\sum\limits_{H_j\subseteq H} \left[f(S\cup H_j+x) - f(S\cup H_j)\right] && \text{(submodularity of $f$)}\\
    &= h(S+x)-h(S).
\end{align*}
Thus $h$ is also submodular.
Finally, we lower bound $h_H(A)$ as below:
\begin{align*}
    h_H(A)&~=~ \frac{1}{2^{|H|}} \sum\limits_{H_j\subseteq H} f(A\cup H_j) \\
    &~=~ \frac{1}{2^{|H|+1}} \sum\limits_{H_j\subseteq H} \left[f(A\cup H_j) +f(A\cup H\setminus H_j) \right] 
    \\
    &~\geq ~\frac{1}{2^{|H|+1}} \sum\limits_{H_j\subseteq H} \left[f(A\cup H)+f(A) \right] &&\text{(submodularity of $f$)}\\
    & ~=~ \frac{1}{2}  \left[f(A\cup H)+f(A) \right].
\end{align*}
\end{proof}

\noindent
Given a cardinality $r$ and a set $H\subseteq N$, we define a $(r-|H|)$-cardinality constraint confined on $N\backslash H$ as $\mathcal{I}_H(r) = \left\{S\subseteq N\setminus H: |S| = r-|H|\right\}$.
Similar to Section~\ref{sec:notation_I}, we provide the following lemma indicating that the auxiliary function $\varphi_{h_H}$ can be well approximated.

\begin{lemma}[\bf{Approximation of $\varphi_{h_H}$}]
\label{lem:hat_varphi_h2}
Let $\varepsilon>0$ and assume $n\in \Omega\left(\varepsilon^{-4}\right)$. 
For a set $H$ with $|H|\geq 3\ln n$, there exists a value oracle $\mathcal{O}$ to a $\left(\alpha=\frac{\varepsilon}{4r\ln r}, \delta=\frac{3}{n^6}, \mathcal{I}_H(r)\right)$-approximation $\widehat{\varphi_{h_H}}$ of $\varphi_{h_H}$, which answering $\mathcal{O}(A)$ queries $\Tilde{f}$'s oracle $O\left(r\varepsilon^{-1}\log^{\frac{5}{2}}n\log^2r\right)$ times for each set $A\in \mathcal{I}_H(r)$.
\end{lemma}

\noindent
The above lemma indicates that for sufficiently large $n$ and $|H|$, we can achieve an almost accurate approximation $\widehat{\varphi_{h_H}}$ of $\varphi_{h_H}$, which queries $\tilde{f}$ only $\tilde{O}(n\varepsilon^{-1})$ times for each $\widehat{\varphi_{h_H}}(A)$.
By Corollay~\ref{cor:main_nls}, we set $\delta = \frac{3}{n^6}$ 
when applying Algorithm~\ref{alg:local_search}.
\begin{proof}
    
For any set $A\subseteq N$, we rewrite $\varphi_{h_H}(A)$ as a linear combination of $f(T)$ according to Definitions \ref{def:auxiliary} and~\ref{def:h_2}:
\[\varphi_{h_H}(A) =\frac{1}{2^{|H|}} \sum\limits_{T\subseteq A}\sum\limits_{H_i\subseteq H}m_{|A|-1,|T|-1}f(T\cup H_i).\]
We will use Algorithm \ref{alg:hat_varphi_h2} to construct the oracle of approximate function $\widehat{\varphi_{h_H}}$. 
Similar to Lemma  \ref{lem:hat_varphi_h},  we first show some bounds of coefficients $m_{|A|-1,|T|-1}$.
With these bounds, we will then prove $\varphi_{h_H}$ and its noisy analogue $\widetilde{\varphi_{h_H}}$ are close to each other by Bernstein’s inequality for sub-exponential variables. 
Finally, we use Hoeffding’s inequality to show that with enough samples, $\widehat{\varphi_{h_H}}$ is close to $\widetilde{\varphi_{h_H}}$ and therefore a good approximation of $\varphi_{h_H}$.

\begin{algorithm}[htp] 
	\caption{Approximation of $\varphi_{h}$} 
	\label{alg:hat_varphi_h2}
	\DontPrintSemicolon
    \SetNoFillComment
    \SetKwInOut{Input}{Input}
		\Input{Noisy oracle $\Tilde{f}$, a set $A\in\mathcal{I}_H(r)$
		and $\varepsilon\in(0, 1/2)$}
		Let $M = 32rH_r^2\log^{\frac{5}{2}}n\cdot\varepsilon^{-1}$ and $s(A) = \sum_{T\subseteq A}m_{|A|-1,|T|-1}$\;
         Construct a distribution $\nu_A(B)$ on set $2^N$ such that
		$\nu_A(B) =\left\{\begin{array}{ll}
     \dfrac{m_{|A|-1,|T|-1}}{2^{|H|}s(A)}&  \mbox{ if there exists $T\subseteq A$ and $H_i\subseteq H$ such that $B= T\cup H_i$} \\
     0 & \mbox{otherwise} 
\end{array} \right.$\;
Sample $M$ subsets $B_1,\cdots,B_M$ with distribution $\nu_A(B)$\;
		 $\widehat{\varphi_{h}}(A)\leftarrow s(A) \sum\limits_{i=1}^M\Tilde{f}(B_i)/M$\;
	\Return $\widehat{\varphi_{h}}(A)$
\end{algorithm}

\begin{lemma}[\bf{Useful bounds of the coefficients}]
\label{lmm:app_card2_bound}
For any $A\subseteq N$ with size $|A|\geq 2$, we have 
\begin{align*}
    &(1)~~\sum\limits_{T\subseteq A}m_{|A|-1,|T|-1} \leq H_{|A|}; && \text{(\cite[Lemma 2]{filmus2014monotone})}\\
    &(2)~~\sum\limits_{T\subseteq A}m^2_{|A|-1,|T|-1}\leq \dfrac{e^2}{(e-1)^2}\dfrac{2}{|A|}\leq 3.
\end{align*}
\end{lemma}
\begin{proof}
It only remains to estimate the bound of $\sum\limits_{T\subseteq A}m^2_{|A|-1,|T|-1}$.
We have
\begin{align*}
    \sum\limits_{T\subseteq A}m^2_{|A|-1,|T|-1}
    ~ =~& \sum\limits_{i=1}^{|A|} \left(\int_0^1\frac{e^p}{e-1}p^{i-1}(1-p)^{|A|-i} dp\right)^2 \binom{|A|}{i} \tag*{(Definition~\ref{def:auxiliary})}\\
    ~\leq~&  \frac{e^2}{(e-1)^2}\sum\limits_{i=1}^{|A|} \left(\int_0^1 p^{i-1}(1-p)^{|A|-i} dp\right)^2 \binom{|A|}{i}\tag*{($p\in [0,1]$)} \\
     ~ =~& \frac{e^2}{(e-1)^2}\sum\limits_{i=1}^{|A|} \left(\frac{\Gamma(i)\Gamma(|A|-i+1)}{\Gamma(|A|+1)}\right)^2 \binom{|A|}{i}\tag*{(Ineq.~\eqref{eq:beta})} \\
    ~=~ &\frac{e^2}{(e-1)^2}\left(\frac{1}{|A|} + \sum\limits_{i=2}^{|A|} \frac{1}{i(|A|-i+1)} \frac{1}{\binom{|A|}{i-1}}  \right)\tag*{($\binom{|A|}{i-1}\geq |A|$)}\\ 
    ~\leq~&  \frac{e^2}{(e-1)^2}\frac{2}{|A|} \leq 3.
\end{align*}
\end{proof}
\noindent To show that $\widehat{\varphi_{h_H}}$ is a $\left(\alpha=\frac{\varepsilon}{4r\ln r}, \delta=\frac{3}{n^6}, \mathcal{I}_H(r)\right)$-approximation of $\varphi_{h_H}$,
it suffices to prove
\[ \mathbb{P}\left[\left|\widehat{\varphi_{h_H}}(A) - \varphi_{h_H}(A)\right| > \frac{\varepsilon}{4rH_r}\cdot  \varphi_{h_H}(A)\right]\leq \frac{3}{n^6}.\]
First, we prove the the following lemma using concentration properties of sub-exponential random variables. 
\begin{lemma}[\textbf{Difference between $\varphi_{h_H}$ and $\widetilde{\varphi_{h_H}}$}]
\label{clm:4.2_2}
If $|H| = \lceil3\ln n\rceil$, we have  
\begin{equation}
\label{eq:tilde_varphi}
    \mathbb{P}\left[|\widetilde{\varphi_{h_H}}(A) - \varphi_{h_H}(A) | > \dfrac{\varepsilon}{8rH_r} \varphi_{h_H}(A)\right]\leq \frac{2}{n^6}.
\end{equation}
\end{lemma}

\begin{proof}
Let $\xi$ be a random variable with a generalized exponential tail distribution $\mathcal{D}$.
The varaiable $\xi-1$ has a finite sub-exponential norm $\kappa$ by Claim~\ref{clm:subexp}.
Let 
$$X_{i,j} = \frac{1}{2^{|H|}}\cdot m_{|A|-1,|T_j|-1} f(T_j\cup H_i)\cdot(\xi_{T_j\cup H_i}-1),$$ 
where $H_i\subseteq H$, $T_j\subseteq A$ and $i\in [2^{|H|}]$, $j\in [2^{|A|}]$. 
Applying Bernstein’s inequality (Lemma~\ref{lem:bern_exp}) gives
\begin{align*}
\resizebox{\textwidth}{!}{$\displaystyle
 \mathbb{P}\left\{\left|\widetilde{\varphi_{h_H}}(A) - \varphi_{h_H}(A) \right|> \dfrac{\varepsilon\varphi_{h_H}(A)}{8rH_r} \right\}
 \leq 2 \exp 
\left[
-c \min\left(
\frac{\left(\dfrac{\varepsilon}{8rH_r}\right)^2\left(\varphi_{h_H}(A)\right)^2}{\sum\limits_{i,j}\left\|X_{i,j}\right\|^2_{\psi_1}},
\frac{\dfrac{\varepsilon}{8rH_r}\varphi_{h_H}(A)}{\max\limits_{i,j}\left\|X_{i,j}\right\|_{\psi_1}}
\right)
\right].$}
\end{align*}
We can bound $\max_{i,j}\left\|X_{i,j}\right\|_{\psi_1}$ by 
     \begin{equation*}
         \max_{i,j}\left\|X_{i,j}\right\|_{\psi_1} \leq \frac{\kappa}{2^{|H|}}
          f(A\cup H)\cdot  \max_{T_j}m_{|A|-1,|T_j|-1} \leq \frac{3}{2^{|H|}}\kappa\cdot f(A\cup H),
     \end{equation*}
where we use the fact that for any $a, t \geq 0$,
$$m_{a, t} = \int_0^1 \frac{e^p}{e-1} p^t(1-p)^{a-t}dp\leq 3.$$
With Lemma~\ref{lmm:app_card2_bound}, $\sum\limits_{i,j}\left\|X_{i,j}\right\|^2_{\psi_1}$ can be bounded by
$$\sum\limits_{i,j}\left\|X_{i,j}\right\|^2_{\psi_1} \leq \frac{1}{2^{|H|}}\sum\limits_{T_j\subseteq A}m^2_{|A|-1,|T_j|-1} \left(f(A\cup H)\right)^2\kappa^2 \leq \frac{3}{2^{|H|}}\kappa^2(f(A\cup H))^2.$$
For a sufficiently large $n\in \Omega(\varepsilon^{-4})$, we have 
$|H| \geq 3\ln n \geq  \ln\left(\frac{36864\kappa^2}{c}r^2H_r^2 \ln n\cdot\varepsilon^{-2}\right)$.
Then by Lemma~\ref{lmm:property_of_h2} and direct calculations, we can conclude that 
\begin{align*}
    \min\left(\frac{\left(\dfrac{\varepsilon}{8rH_r}\right)^2\left(\varphi_{h_H}(A)\right)^2}{\sum\limits_{i,j}\left\|X_{i,j}\right\|^2_{\psi_1}},
\frac{\dfrac{\varepsilon}{8rH_r}\varphi_{h_H}(A)}{\max\limits_{i,j}\left\|X_{i,j}\right\|_{\psi_1}}\right) \geq \frac{6}{c}\ln n.
\end{align*}
Therefore Ineq.~\eqref{eq:tilde_varphi} holds.
\end{proof}
\noindent
Next we prove the following lemma using Hoeffding’s inequality (Lemma~\ref{lmm:Hoeffding}).
\begin{lemma}[\textbf{Difference between $\widehat{\varphi_{h_H}}$ and $\widetilde{\varphi_{h_H}}$}]
\label{clm:4.2_1}
If $n\in \Omega(\varepsilon^{-4})$, we have  
\begin{equation*}
    \mathbb{P}\left[|\widehat{\varphi_{h_H}}(A) - \widetilde{\varphi_{h_H}}(A) | > \dfrac{\varepsilon}{8rH_r} \varphi_{h_H}(A)\right]\leq \frac{1}{n^6}.
\end{equation*}
\end{lemma}

\begin{proof}
To use Hoeffding's inequality, we need a bound $b$ for noise multipliers.
If the noise distribution $\mathcal{D}$ has bounded support, this bound is natural.
When the support of $\mathcal{D}$ is not bounded, we let \[b = \frac{16}{c_0}\ln\left(n\cdot c_0^{-\frac{1}{4}}\right).\]
Recall that we sample $M$ sets $B_1,\dots,B_M$ to estimate $\widehat{\varphi_{h}}(A)$ in Algorithm~\ref{alg:hat_varphi_h2}.
We say that event $\mathcal{E}_A$ happens if for all $i\in [M]$, we have $\xi_{B_i} \leq b$.
If the support of $\mathcal{D}$ is bounded, event $\mathcal{E}_A$ trivially holds.
Otherwise, for a distribution with unbounded support, recall that when 
$$x\geq \max\left\{1, \left(\frac{2\sum_{i=1}|c_i|}{c_0}\right)^{\frac{1}{\gamma_0-\gamma_1}}\right\},$$
we have the probability density $\rho(x) \leq e^{-\frac{1}{2}c_0x^{\gamma_0}} \leq e^{-\frac{1}{2}c_0x}$.
When $n$ is sufficiently large, we have
$$\mathbb{P}\left[\xi\geq b\right] \leq \int_{b}^{\infty} \exp\left(-\frac{1}{2}c_0x\right)dx 
\leq \frac{2}{c_0}\exp\left(-\frac{1}{2}c_0b\right)= \frac{2}{n^8}.$$
Thus by taking a union bound, event $\mathcal{E}_A$ happens with probability at least $1-{2M}/{n^8}$.
Conditioned on event $\mathcal{E}_A$, consider the random variables $X_i = s(A)\Tilde{f}(B_i)$. 
We bound $X_i$ by
\begin{align*}
 X_i 
  & ~=~ s(A)\xi_{B_i} f(B_i) \tag*{(Definition \ref{def:noise})}\\
  & ~\leq~ H_{r} b f(B_i) \tag*{(Lemma \ref{lmm:app_card2_bound} with $|A|\leq r$ and event $\mathcal{E}_A$)}\\
  & ~\leq~ H_{r} b f(A\cup H) .\tag*{(monotonicity of $f$)}
\end{align*}
Thus applying Hoeffding’s inequality (Lemma \ref{lmm:Hoeffding}), we obtain
\begin{align*}
 & \mathbb{P}\left[|\widehat{\varphi_{h_H}}(A)- \widetilde{\varphi_{h_H}}(A)|> \dfrac{\varepsilon}{8rH_r} \varphi_{h_H}(A)\right] 
    ~\leq ~ \frac{2}{n^8},
\end{align*}
which holds for a sufficiently large $n$ such that $\ln^2n \geq b$.
Therefore, $|\widehat{\varphi_h}(A)-\widetilde{\varphi_h}(A)|> \frac{\varepsilon}{8rH_r}\varphi_h(A)$ holds with probability at most $$\frac{2M}{n^8} +\left(1-\frac{2M}{n^8}\right)\frac{2}{n^8} = O(n^{-6}).$$
\end{proof}
\noindent
Combining Lemma \ref{clm:4.2_2} and \ref{clm:4.2_1}, we complete the proof of Lemma~\ref{lem:hat_varphi_h2}.

\end{proof}
\subsubsection{Algorithm for Theorem~\ref{thm:cardi2}}
\noindent
We present Algorithm~\ref{alg:local_search_monotone2} that mainly contains a local search phase (Line 3).
The main difference from Algorithm~\ref{alg:local_search_monotone} is that Algorithm \ref{alg:local_search_monotone2}
arbitrarily select a redundant set $H$ (Line 1), and use this set to construct a smoothing surrogate function $h_{H}$ for local search (Line 3). 
Then Algorithm~\ref{alg:local_search_monotone2} returns  $S_L\cup H$ directly (Line 4) without a $\Tilde{f}$-maximization phase. 

\begin{algorithm}[htp!] 
	\caption{Noisy local search under large cardinality constraint} 
	\label{alg:local_search_monotone2}
	\DontPrintSemicolon
    \SetNoFillComment
    \SetKwInOut{Input}{Input}
	\Input{a value oracle to $\Tilde{f}$, budget $r\in  \Omega\left(n^{\frac{1}{3}}\right)$, and $\varepsilon\in (0,1/2)$}
    Arbitrarily select a subset $H\subseteq N$ with $|H|= \lceil 3\ln n \rceil$ \;
    Let $\widehat{\varphi_{h_{H}}}$ be a $\left(\alpha =\frac{ \varepsilon}{4r\ln r},\delta =\frac{3}{n^6},\mathcal{I}_{H}(r)\right)$-approximation of $\varphi_{h_{H}}$ as in Lemma \ref{lem:hat_varphi_h2} \;
	$S_L\leftarrow \mathtt{NLS}\left(\widehat{\varphi_{h_{H}}},\mathcal{I}_{H}(r),\Delta = \frac{\varepsilon}{4r\ln r}\right)$ \Comment*[r]{Local search phase}
    \Return $S_L\cup H$ 
\end{algorithm}

\subsubsection{Proof of Theorem~\ref{thm:cardi2}}

Again, we first analyze the query complexity and then prove the approximation performance. 

\paragraph{Query complexity of Algorithm~\ref{alg:local_search_monotone2}.}
Since $n\in \Omega\left(\frac{1}{\varepsilon^4}\right)$, we have $I\in O(n^2)$ in Corollary~\ref{cor:main_nls}.
Hence Line 4 calls the oracle of $\widehat{\varphi_{h_H}}$  at most $O(rn^3)$ times by Corollary~\ref{cor:main_nls}, and each call to $\widehat{\varphi_{h_H}}$ queries $\Tilde{f}$ at most $O\left(r\varepsilon^{-1}\log^{\frac{5}{2}}n\log^2r\right)$ times by Lemma~\ref{lem:hat_varphi_h2}.
Thus we query $\Tilde{f}$ at most $O(n^6\varepsilon^{-1})$ times.

\paragraph{Approximation ratio and success probability of Algorithm~\ref{alg:local_search_monotone2}.}
Corollary \ref{cor:main_nls}, we have
    $h_H(S_L)\geq \left(1-\frac{1}{e}-\varepsilon\right)\cdot\max_{S\subseteq \mathcal{I}_H(r)} h_H(S)$ with probability $1-O\left(\frac{1}{n^2}\right)$.
Furthermore, $f(S_L\cup H)\geq h_H(S_L)$ follows from the monotonicity of $f$.
Then we can demonstrate the approximation performance of Algorithm~\ref{alg:local_search_monotone2} by the following lemma.
\begin{lemma}
\label{clm:card2_split}
Let $O^\star = \mathop{\arg\max}_{S\in \mathcal{I}(r)}f(S)$. 
We have
$\max_{S\in \mathcal{I}_H(r)}h_H(S)\geq \left(1-\frac{|H|}{r}\right)f(O^\star)$.
\end{lemma}
\noindent
Roughly speaking, there exists a subset $A\subseteq O^\star$ of size $|H|$ such that
$f(O^\star\backslash A)\geq \left(1-\frac{|H|}{r}\right)f(O^\star)$ by submodularity of $f$.
The claim follows from $\max_{S_L\in \mathcal{I}_H(r)}h_H(S)\geq h(O^\star\backslash A)\geq f(O^\star\backslash A)$.
\begin{proof}
    We first show that a random set $A$, which is uniformly selected from all subsets of $N$ with size $|H|$, satisfies that
\[\mathop{\mathbb{E}}\left[f(O^\star\setminus A)\right]\geq \left(1-\frac{|H|}{r}\right)f(O^\star).\]
To prove this, we index $O^\star$ as $\{o_1,\cdots,o_r\}$ and denote by $I_{A}$ the indexes of elements in $O^{\star}\cap A$. 
Then we decompose $f(O^\star)$ as below: 
\[f(O^\star)=\sum\limits_{i=1}^rf(o_i\mid o_t,~t\in [i-1]),\]
where $o_0$ is defined to be $\varnothing$.
Similarly, we decompose $f(O^\star\setminus H)$ as 
\begin{align*}
f(O^\star\setminus A)~=~&\sum\limits_{i\in [r]\setminus I_A}f(o_i\mid o_t,~t\in [i-1] \mbox{ and } t \not\in I_A)\\ 
~\geq~& \sum\limits_{i\in [r]\setminus I_A}f(o_i\mid o_t,~t\in [i-1] ).\tag*{(submodularity of $f$)}\\
\end{align*} 
Taking an expectation gives
\begin{align*}
\mathop{\mathbb{E}}\left[f(O^\star\setminus A)\right]
~\geq~& \mathop{\mathbb{E}}\left[\sum\limits_{i\in [r]\setminus I_A}f\left(o_i\mid o_t,~t\in [i-1] \right)\right]\\
~=~& \left(1-\frac{|H|}{r}\right)\left(\sum\limits_{i=1}^rf(o_i\mid o_t,~t\in [i-1])\right) = \left(1-\frac{|H|}{r}\right)f(O^\star).
\end{align*}
Thus there must exist a set $A_0$ such that 
$f(O^\star\setminus A_0)  
\geq 
\left(1-\frac{|H|}{r}\right)f(O^\star)$.
Therefore, we have
\[\max\limits_{S\in \mathcal{I}_H(r)}h_H(S)\geq h_H(O^\star \setminus A_0)\geq f(O^\star \setminus A_0)\geq \left(1-\frac{|H|}{r}\right)f(O^\star).\]
\end{proof}
\noindent
Combining the fact that $f(S_L\cup H)\geq \left(1-\frac{1}{e}-\varepsilon\right)\max_{S\subseteq \mathcal{I}_H(r)} h_H(S)$ with Lemma \ref{clm:card2_split}, we complete the proof of Theorem~\ref{thm:cardi2}.

\section{Our algorithms and main theorems for general matroid constraints}
\label{sec:matroid}

In this section, we consider Problem~\ref{problem:noisy_submodular} under matroid constraints.
Similarly, for different ranges of matroid ranks, we prove there are  $(1-\frac{1}{e}-O(\varepsilon))/2$-approximate algorithms (Theorems \ref{thm:matroid1} and \ref{thm:matroid2}). 
W.l.o.g., we assume that all single elements are feasible for the given matroid constraints.
\subsection{Algorithmic results for matroid constraints with rank \texorpdfstring{$r\in \Omega\left(\frac{1}{\varepsilon}\right) \cap O\left(n^{\frac{1}{3}}\right)$}{Lg}}

In this section, we generalize Algorithm~\ref{alg:local_search_monotone} to deal with general matroid constraints, which results in an algorithm (Algorithm~\ref{alg:local_search_monotone3}) achieving the performance stated in Theorem \ref{thm:matroid1}.

\begin{theorem}[\textbf{Algorithmic results for matroid constraints with rank $r\in  \Omega\left(\frac{1}{\varepsilon}\log\left(\frac{1}{\varepsilon}\right)\right) \cap O\left(n^{\frac{1}{3}}\right) $}]
\label{thm:matroid1}
  Let $\varepsilon > 0$ is sufficiently small and assume $n \in \Omega\left(\exp\left(\left(\frac{1}{\varepsilon}\right)^{O(1)}\right)\right)$ is sufficiently large. 
    For any $r\in \Omega(\varepsilon^{-1}\log(\varepsilon^{-1})) \cap O\left(n^{\frac{1}{3}}\right)$, there exists an algorithm that returns a $\left(\left(1-\frac{1}{e}\right)/2-O(\varepsilon)\right)$-approximation for Problem~\ref{problem:noisy_submodular} under a matroid constraint $\mathcal{I}(\mathcal{M})$ with rank $r$, with probability at least $1-O\left(\varepsilon^4\right)$ and query complexity at most $O\left(r^2\log ^2r\cdot n^{\frac{3}{2}}\varepsilon^{-1}\max\{r, \log n\}\right)$ to $\tilde{f}$.
\end{theorem}

\noindent
The main difference from Theorem \ref{thm:cardi} is that the approximation ratio is  $\left(1-\frac{1}{e}\right)/2-O(\varepsilon)$ instead of $1-\frac{1}{e}-O(\varepsilon)$.
The reason is that to maintain the feasibility of the output, we may not be able to add an element to $S_L$ as in Line 3 of Algorithm~\ref{alg:local_search_monotone}.
To address this issue, we design a comparison procedure (Lines 4-8 of Algorithm~\ref{alg:local_search_monotone3}) to evaluate the value of $f(S_L)$ and $f(S_M\setminus S_L)$ and output the larger one, which results in a loss on the approximation ratio.
Moreover, the failure probability of the comparison procedure is upper bounded by $O(\varepsilon^4)$.

\subsubsection{Useful notations and useful facts for Theorem~\ref{thm:matroid1}}
\label{sec:surrogate_III}
The local search procedure in this subsection is still based on the smoothing surrogate function $h$ in Definition \ref{def:h_1}.
Similar to Lemma~\ref{lem:hat_varphi_h}, we can construct a $\left(\alpha,\delta, \mathcal{I}(r-1)\cap \mathcal{I}(\mathcal{M})\right)$-approximation of auxiliary function $\varphi_h$.
\begin{lemma}[\textbf{Approximation of $\varphi_{h}$}]
\label{lem:hat_varphi_h_3}
%
Let $\alpha,\delta \in (0,1/2)$ and assume $n\in \Omega(\alpha^{-2}\log(\delta^{-1}))$. 
There exists a value oracle $\mathcal{O}$ to an $\left(\alpha,\delta, \mathcal{I}(r-1)\cap \mathcal{I}(\mathcal{M})\right)$-approximation $\widehat{\varphi_{h}}$ of $\varphi_{h}$, which answering $\mathcal{O}(A)$ queries at most $O\left(\log r\cdot n^{\frac{1}{2}}\max\{r, \log n\}\right)$ times to $\tilde{f}$ for each set $A\in\mathcal{I}(r-1)\cap \mathcal{I}(\mathcal{M})$.
\end{lemma}
\noindent
Compared with the $\left(\alpha,\delta, \mathcal{I}(r-1)\right)$-approximation in Lemma~\ref{lem:hat_varphi_h}, an $\left(\alpha,\delta, \mathcal{I}(r-1)\cap \mathcal{I}(\mathcal{M})\right)$-approximation requires a smaller estimation error $\alpha\cdot\max\limits_{S\in \mathcal{I}(r-1)\cap\mathcal{I}(\mathcal{M})}\varphi_h(S)$ than $\alpha\cdot\max\limits_{S\in \mathcal{I}(r-1)}\varphi_h(S)$.
However, the proof of Lemma~\ref{lem:hat_varphi_h} remains valid for Lemma~\ref{lem:hat_varphi_h_3}.
This is because for any set $A\in \mathcal{I}(r-1)\cap\mathcal{I}(\mathcal{M})$, we can still bound $f(T)$ for all $T\in \mathcal{L}_A$ by $2\cdot\max_{S\in\mathcal{I}(r-1)\cap \mathcal{I}(\mathcal{M})}\varphi_h(S)$ as in Eq.~\eqref{eq:max}.
Thus the concentration results (Lemma~\ref{lmm:app_diff_phi_hat1} and~\ref{lmm:app_diff_tilde_hat_1}) in the proof hold as well.

Besides $\varphi_h$ (Definition~\ref{def:auxiliary}) guiding the local search, we introduce another auxiliary function $\widehat{f}_0$ in this subsection, which will be used to compare the values of sets at the final step of Algorithm~\ref{alg:local_search_monotone3}.
\begin{definition}[\bf{Comparison auxiliary function}]
\label{def:f_0}
For any set $S\subseteq N$, we define the comparison auxiliary function $f_0(S)$ as the expectation of $f(S-e)$ over a random element $e\in S$, i.e., 
$$f_0(S) = \frac{1}{|S|}\sum_{e\in S}f(S-e).$$
\end{definition}
\noindent
The comparison auxiliary function $f_0$ is close to $f$ in the sense that the following lemma holds.
\begin{lemma}[\textbf{Bounds of $f_0$}]
\label{lem:bounds_f0}
For any set $S\subseteq N$, we have
$$\left(1-\frac{1}{|S|}\right)f(S)\leq f_0(S)\leq f(S).$$
\end{lemma}
\begin{proof}
From monotonicity of $f$, we have $f_0(S) \leq
 f(S)$ immediately.
We index the elements of $S$ as $s_i$, where $i\in[|S|]$. Then
\begin{align*}
f_0(S) &~=~ \frac{1}{|S|}\sum_{e\in S}\left[ f(S)-f(e\mid S-e)\right] && \text{(Definition~\ref{def:f_0})}\\
&~ = ~f(S) - \frac{1}{|S|}\sum_{e\in S}f(e\mid S-e)\\
&~\geq~ f(S) - \frac{1}{|S|}\sum_{i \in [|S|]}f(s_i\mid s_1,s_2,\cdots,s_{i-1}) && \text{(submodularity of $f$)}\\
& ~= ~\left(1-\frac{1}{|S|}\right) f(S). && \text{($\sum_{i \in [|S|]}f(s_i\mid s_1,\cdots,s_{i-1}) = f(S)$)}
\end{align*}
\end{proof}
\noindent
Note that $f_0$ is also implicitly constructed since we only have a value oracle to $\tilde{f}$ instead of $f$.
We denote $f_0$'s noisy analogue as $\widetilde{f_0}$, i.e.,
$$\widetilde{f_0}(S) = \frac{1}{|S|}\sum_{e\in S}\Tilde{f}(S-e).$$
As $f_0(S)$ is based on an averaging set of size $|S|$, the following lemma indicates that $f_0$ and $\widetilde{f_0}$ are close to each other when $|S|$ is sufficiently large.
\begin{lemma}[\bf{Concentration property of $f_0$}]
\label{lmm:matroid1_comp} 
Let $\varepsilon, \delta\in(0,1/2)$ and suppose that $|S| \geq \kappa\varepsilon^{-1}\ln(2\delta^{-1})$, where $\kappa$ is the sub-exponential norm of the noise multiplier.
For any subset $S\subseteq N$, we have 
\[\mathbb{P}\left[\left|\widetilde{f}_0(S)-f_0(S)\right|> \varepsilon \cdot f_0(S) \right] \leq \delta.\]
\end{lemma}
\begin{proof}
For any set $T\in\{S-e:e\in S\}$, let $X_T = \xi(T)-1$ and $a_T = \frac{1}{|S|}\cdot f(T)$.
Recall that $X_T$ is a sub-exponential random variable with norm $\kappa$ by Claim~\ref{clm:subexp}.
Applying Bernstein's inequality for sub-exponential variables (Lemma~\ref{lem:bern_exp}) gives
$$\mathbb{P}\left[\left|\widetilde{f}_0(S)-f_0(S)\right|> \varepsilon f_0(S) \right] \leq 2 \exp 
\left[
-c \min\left(
\frac{\varepsilon^2 |S|^2}{\kappa^2}, \frac{\varepsilon|S|}{\kappa}
\right)\right].$$
By assumption that 
$|S|\geq \kappa\varepsilon^{-1}\ln(2\delta^{-1})$, the probability above is at most $\delta$.
\end{proof}

\subsubsection{Algorithm for Theorem~\ref{thm:matroid1}}
We present Algorithm~\ref{alg:local_search_monotone3} that is a variant of Algorithm~\ref{alg:local_search_monotone}.
The main difference is that Algorithm~\ref{alg:local_search_monotone3} contains a comparison phase (Line 4-7).
Since $S_M$ may not be an independent set defined by $\mathcal{I}(\mathcal{M})$, Algorithm~\ref{alg:local_search_monotone3} compares $S_L$ returned by local search with the element $S_M\backslash S_L$ obtained at the $\Tilde{f}$-maximization phase and returns the one with a larger $\widetilde{f_0}$ value.

\begin{algorithm}[htp]
 \caption{Noisy local search subject to matroids with small ranks}
  \label{alg:local_search_monotone3}
\DontPrintSemicolon
\SetNoFillComment
  \SetKwInOut{Input}{Input}
  \Input{a value oracle to $\Tilde{f}$, rank $r\in \Omega\left(\frac{1}{\varepsilon}\log\left(\frac{1}{\varepsilon}\right)\right) \cap O\left(n^{\frac{1}{3}}\right)$ and $\varepsilon\in (0,1/2)$}
        Let $\widehat{\varphi_h}$ be a $\left(\frac{\varepsilon}{4r\ln r}, \frac{1}{(I+1)(r-1)n^2}, \mathcal{I}(r-1)\cap \mathcal{I}(\mathcal{M})\right)$-approximation of $\varphi_h$ as in Lemma~\ref{lem:hat_varphi_h_3} \; 
    	$S_L\leftarrow\mathtt{NLS}\left(\widehat{\varphi_h}, \mathcal{I}(r-1)\cap \mathcal{I}(\mathcal{M}), \Delta = \frac{\varepsilon}{4r\ln r}\right)$ \Comment*[r]{Local search phase}
		$S_M\leftarrow S_L +\mathop{\arg\max}\limits_{e\in N\backslash S_L} \Tilde{f}(S_L+e)$\Comment*[r]{$\Tilde{f}$-maximization phase}
		\If{$\widetilde{f}_0(S_L)\geq \frac{1}{2}\widetilde{f}_0(S_M)$}{
         \Return $S_L$
        \Comment*[r]{Comparison phase}}
        \Else{
         \Return $S_M\backslash S_L$
        }
\end{algorithm}
\subsubsection{Proof of Theorem~\ref{thm:matroid1}}
We first analyze the query complexity and then prove the approximation performance of Algorithm~\ref{alg:local_search_monotone3}. 
\paragraph{Query complexity of Algorithm~\ref{alg:local_search_monotone3}.}
We call the noisy oracle $\Tilde{f}$ only $(2r-1)$ times at the comparison phase.
The query complexity of Algorithm~\ref{alg:local_search_monotone3} is dominated by the number of calls it makes at the local search phase and therefore of the same order $\Tilde{O}(r^2n^{\frac{3}{2}}\varepsilon^{-1})$ as Algorithm~\ref{alg:local_search_monotone}.

\paragraph{Approximate ratio and success probability of Algorithm~\ref{alg:local_search_monotone3}.}
Recall that $O_f\in \mathop{\arg\max}\limits_{S\in \mathcal{I}(\mathcal{M})}f(S)$ denote an optimal solution to $f$.
Similar to the analysis of Algorithm~\ref{alg:local_search_monotone}, it can be shown that we obtain a set $S_M$ at the end of the $\Tilde{f}$-maximization phase such that with probability $1-O\left(1/\log n\right)$,
\begin{equation}
\label{eq:f_SM}
    f(S_M) \geq \left(1-\frac{1}{e}-O(\varepsilon)\right)f(O_{f}).
\end{equation}
Although $S_M$ may not be an independent set, we can decompose $S_M$ into two feasible parts $S_L\in \mathcal{I}(\mathcal{M})$ and $S_M\backslash S_L\in \mathcal{I}(\mathcal{M})$ and output one of them.
The approximate ratio of the output set is guaranteed by the following lemma.

\begin{lemma}
\label{clm:compare_m}
We assume a sufficient small $\varepsilon$ and suppose that $r\in \Omega\left(\varepsilon^{-1}\log(\varepsilon^{-1})\right)$.
Let $S_R$ denote the set returned by Algorithm~\ref{alg:local_search_monotone3}.
With probability at least $1-2\varepsilon^4$, we have 
\begin{equation*}
    f(S_R)\geq \left(\frac{1}{2}-O(\varepsilon)\right)f(S_M).
\end{equation*}
\end{lemma}
\begin{proof}
By Lemma~\ref{lmm:matroid1_comp} with $\delta = \varepsilon^{4}$, we have
\begin{equation}
\label{eq:event_c}
    \left|\widetilde{f}_0(S)-f_0(S)\right|\leq  \varepsilon f_0(S)
\end{equation}
holds for both $S_L$ and $S_M$ with probability at least $1-2\varepsilon^{4}$.
Suppose this is true.

If $\widetilde{f}_0(S_L)\geq \frac{1}{2}\widetilde{f}_0(S_M)$, then $S_R$ is $S_L$.
We translate the inequality into that of $f_0$ by Ineq.~\eqref{eq:event_c}:
$$f_0(S_L)\geq  \frac{\widetilde{f}_0(S_L)}{1+\varepsilon}\geq  \frac{\widetilde{f}_0(S_M)}{2(1+\varepsilon)}\geq \frac{1}{2}\cdot \frac{1-\varepsilon}{1+\varepsilon}\cdot f_0(S_M).$$
Then from Lemma~\ref{lem:bounds_f0}, we have
$$f(S_L) \geq f_0(S_L)\geq \frac{1}{2}\cdot \frac{1-\varepsilon}{1+\varepsilon}\cdot f_0(S_M) \geq \frac{1}{2}\cdot \frac{1-\varepsilon}{1+\varepsilon}\left(1-\frac{1}{r}\right)f(S_M).$$
For any $r>1/\varepsilon$, the inequality above implies that
$$f(S_L) \geq \left(\frac{1}{2}-2\varepsilon\right)f(S_M).$$

Algorithm~\ref{alg:local_search_monotone3} returns $S_M\setminus S_L$ if $\widetilde{f}_0(S_L)< \frac{1}{2}\widetilde{f}_0(S_M)$.
Similarly, we convert this by Ineq.~\eqref{eq:event_c} into the following inequality:
$$f_0(S_L) \leq \frac{\widetilde{f_0}(S_L)}{1-\varepsilon}\leq \frac{1}{2}\cdot\frac{\widetilde{f_0}(S_M)}{1-\varepsilon} \leq \frac{1}{2}\cdot \frac{1+\varepsilon}{1-\varepsilon}f_0(S_M).$$
By Lemma~\ref{lem:bounds_f0},
$$f(S_L) \leq \frac{r-1}{r-2}f_0(S_L) < \frac{1}{2}\cdot \frac{1+\varepsilon}{1-\varepsilon}\cdot\frac{r-1}{r-2}f_0(S_M) \leq \frac{1}{2}\cdot \frac{1+\varepsilon}{1-\varepsilon}\cdot\frac{r-1}{r-2}f(S_M),$$
which is directly followed by 
$$f(S_L) < \frac{1}{2}(1+5\varepsilon)f(S_M)$$ if $\varepsilon<1/3$ and $r\geq 2/\varepsilon$.
Since submodularity of $f$ indicates that $f(S_L)+f(S_M\setminus S_L)\geq f(S_M)$, we have 
\[f(S_M\setminus S_L)\geq \left(\frac{1}{2}-3\varepsilon\right)f(S_M).\]
\end{proof}
\noindent
Combining Ineq.~\eqref{eq:f_SM} and Lemma \ref{clm:compare_m}, we prove that Algorithm~\ref{alg:local_search_monotone3} achieves a $((1-1/e)/2-O(\varepsilon))$-approximation with probability $1-O(\varepsilon^4)$.

\subsection{Algorithmic results for matroid constraints with rank  \texorpdfstring{$r\in \Omega\left(n^{\frac{1}{3}}\right)$}{Lg}}
\label{sec:cardinality2}

We present in this subsection an algorithm (Algorithm~\ref{alg:local_search_monotone4}) and its analysis (Theorem~\ref{thm:matroid2}) that can be applied to $r\in \Omega\left(n^{\frac{1}{3}}\right)$ when $n$ is sufficiently large.

\begin{theorem}[\textbf{Algorithmic results for matroid constraints with rank $r\in  \Omega\left(n^{\frac{1}{3}}\right) $}]
\label{thm:matroid2}
    Let $\varepsilon>0$ and assume $n \in \Omega\left(\frac{1}{\varepsilon^4}\right)$ is sufficiently large. 
    For $r\in \Omega\left(n^{\frac{1}{3}}\right)$, there exists an algorithm returning a $\left(\left(1-\frac{1}{e}\right)/2-O(\varepsilon)\right)$-approximation for  Problem~\ref{problem:noisy_submodular} under a matroid constraint $\mathcal{I}(\mathcal{M})$ with rank $r$, with probability $1-O\left(\frac{1}{n^2}\right)$ and query complexity $O(n^6\varepsilon^{-1})$ to $\Tilde{f}$.
    Specifically, for a strongly base-orderable matroid (Definition~\ref{def:strongly}), the approximation ratio can be improved to $1-\frac{1}{e}-O(\varepsilon)$.
\end{theorem}

\noindent
The main difference from Theorem \ref{thm:cardi2} is again the approximation ratio.
%
For maintaning the feasibility, we divide an arbitrary base into two equal-size subsets, apply Algorithm~\ref{alg:local_search_monotone2} to both subsets, and output the larger one.
Consequently, we achieve a $\left(\left(1-\frac{1}{e}\right)/2-O(\varepsilon)\right)$-approximation ratio.
Specifically, for a strongly base-orderable matroid, we can use its good exchangable property to maintain the feasibility without sacrificing the approximation ratio; see Section~\ref{app:matroid3} for a complete proof.

\subsubsection{Useful facts for Theorem~\ref{thm:matroid2}}
\label{sec:matroid2_notation}
Our algorithm and analysis in this subsection are based on the smoothing surrogate function $h_H$ in Definition \ref{def:h_2} and comparison auxiliary function $f_0$ in Definition~\ref{def:f_0}.
We introduce a contraction matroid $\mathcal{I}_H(\mathcal{M})$ confined on $N\backslash H$.
%
Any feasible set $S\in\mathcal{I}_H(\mathcal{M})$ can be converted to an independent set subject to $\mathcal{I}(\mathcal{M})$ by adding $H$.

\begin{definition}[\bf{Contraction matroid}]
\label{def:contraction_matroid}
For any subset $H$ and matroid constraint $\mathcal{I}(\mathcal{M})$, we define contraction matroid over $N\setminus H$ as the following:
\[\mathcal{I}_H(\mathcal{M})= \{S \subseteq N\setminus H : S \cup H \in \mathcal{I}(\mathcal{M})\}.\]
\end{definition}

\noindent
Similar to Lemma~\ref{lem:hat_varphi_h2}, the following lemma indicates that the auxiliary function $\varphi_{h_H}$ can be well approximated.
\begin{lemma}[\bf{Approximation of $\varphi_{h_H}$}]
\label{lem:hat_varphi_h4}
Let $\varepsilon>0$ and for sufficiently large $n\in \Omega(\frac{1}{\varepsilon^4})$. 
If  $|H|\geq 3\ln n$, there exists a value oracle $\mathcal{O}$ to an $(\frac{\varepsilon}{4r\ln r},\frac{3}{n^6},\mathcal{I}_H(\mathcal{M}))$-approximation $\widehat{\varphi_{h_H}}$ of $\varphi_{h_H}$, in which answering $\mathcal{O}(A)$ queries at most $\Tilde{f}$'s oracle $O(r\varepsilon^{-1}\log^{\frac{5}{2}}n\log^2r)$ times of $\Tilde{f}$ for each set $A\in \mathcal{I}$.
\end{lemma}

\noindent
The proof of Lemma~\ref{lem:hat_varphi_h2} actually demonstrates a stricter claim that for any set $A\in \mathcal{I}_{H}(r)$,
$\mathbb{P}\left[\left|\widehat{\varphi_h}(A) - \varphi_h(A)\right| > \frac{\varepsilon}{4r\ln r} \varphi_{h_H}(A)\right]\leq \frac{3}{n^6}$.
Since $\mathcal{I}_H(\mathcal{M}) \subseteq \mathcal{I}_H(r)$, 
the proof remains valid for Lemma~\ref{lem:hat_varphi_h4}.

\subsubsection{Algorithm for Theorem~\ref{thm:matroid2}}
\begin{algorithm}[htp!] 
	\caption{Noisy local search subject to matroids with large ranks} 
	\label{alg:local_search_monotone4}
	\DontPrintSemicolon
    \SetNoFillComment
    \SetKwInOut{Input}{Input}
	\Input{a value oracle to $\Tilde{f}$, rank $r\in  \Omega\left(n^{\frac{1}{3}}\right)$, and $\varepsilon\in (0,1/2)$}
     Arbitrarily select a basis $B_0\subseteq N$ and spilt $B_0$ into two parts $H_1,H_2$ with size $\lfloor\frac{r}{2}\rfloor$ or $\lfloor\frac{r}{2}\rfloor+1$ \\
     
	\For{$t =1,2$} {
		Let $\widehat{\varphi_{h_{H_t}}}$ be a $(\alpha = \frac{\varepsilon}{4r\ln r},\delta =\frac{3}{n^6},\mathcal{I}_{H_t}(\mathcal{M}))$-approximation of $\varphi_{h_{H_t}}$ as in Lemma \ref{lem:hat_varphi_h4}  \;
		$S_t\leftarrow \mathtt{NLS}\left(\widehat{\varphi_{h_{H_t}}},\mathcal{I}_{H_t}(\mathcal{M}),\Delta= \frac{\varepsilon}{4r\ln r}\right)$ \Comment*[r]{Local search phase}
		}
	Let $i^\star \leftarrow \mathop{\arg\max}\limits_{t \in \{1,2\}}\widetilde{f_0}(S_t\cup H_t)$ \Comment*[r]{Comparison phase}
		\Return $S_{i^\star}\cup H_{i^\star}$ 
\end{algorithm}

\noindent
We present Algorithm~\ref{alg:local_search_monotone4} that contains two phases: a local search phase (Lines 2-4) and a comparison phase (Line 5).
Algorithm~\ref{alg:local_search_monotone4} arbitrarily splits a basis (maximal independent set) into two parts of almost the same size, and it then grows each one of them into a basis using Algorithm \ref{alg:local_search}, respectively. 
Finally the algorithm compares two solutions $S_t\cup H_t\ (t=1,2)$ and outputs the better one in terms of $\widetilde{f_0}$. 

\subsubsection{Proof of Theorem~\ref{thm:matroid2}}
To prove Theorem~\ref{thm:matroid2}, we analyze the query complexity and approximation performance of Algorithm~\ref{alg:local_search_monotone3} as below.
\paragraph{Query Complexity of Algorithm~\ref{alg:local_search_monotone4}.} 
%
%
The number of calls to $\Tilde{f}$ during local search phase dominates the query complexity of Algorithm~\ref{alg:local_search_monotone4} as
it queries only $2r$ times at the comparison phase.
Algorithm~\ref{alg:local_search_monotone4} runs the local search procedure twice, and thus the query complexity of Algorithm~\ref{alg:local_search_monotone4} is of the same order $O(n^6\varepsilon^{-1})$ as that of Algorithm~\ref{alg:local_search_monotone2}.


%
\paragraph{Approximation performance analysis of Algorithm~\ref{alg:local_search_monotone4}.}
%
%

Similar to the analysis of Algorithm~\ref{alg:local_search_monotone2},
for $t\in\{1,2\}$, Corollary \ref{cor:main_nls} implies that with probability $1-O\left(\frac{1}{n^2}\right)$,  
\begin{equation}
\label{eq:f_St_cup_Ht}
f(S_t\cup H_t) \geq \left(1-\frac{1}{e}-\varepsilon\right)\max\limits_{S\subseteq \mathcal{I}_{H_t}(\mathcal{M})} h_{H_t}(S).   
\end{equation}
Recall that $O^{\star} \in \mathop{\arg\max}_{S\in \mathcal{I}(\mathcal{M})}f(S)$ denotes an optimal solution.
The following lemma relates $\max_{S\subseteq \mathcal{I}_{H_t}(\mathcal{M})} h_{H_t}(S_t)$ to $f(O^{\star})$.
\begin{lemma}
\label{clm:matroid2_split}
$f(O^{\star}) \leq \max\limits_{S\subseteq \mathcal{I}_{H_1}(\mathcal{M})} h_{H_1}(S_1) + \max\limits_{S\subseteq \mathcal{I}_{H_2}(\mathcal{M})} h_{H_2}(S_2)$.
\end{lemma}
%
%
%
%
%
\begin{proof}
We first introduce a lemma concerning the structure of matroids.
\begin{lemma}[See~\cite{greene1973multiple}]
Given two bases $B_1,\ B_2$ of a matroid $\mathcal{M}$ a partition $B_1 = X_1 \cup Y_1$ , there is a partition $B_2 = X_2 \cup Y_2$ such that $X_1 \cup Y_2$ and $X_2\cup  Y_1$ are both bases of $\mathcal{M}$.
\end{lemma}
The lemma above indicates that there is a partition of $O^{\star}=O_1\cup O_2$ such that $O_t\cup H_t\in\mathcal{I}(\mathcal{M})$ for $t\in \{1,2\}$.
%
Hence we have
\begin{align*}
  \max\limits_{S\subseteq \mathcal{I}_{H_1}(\mathcal{M})} h_1(S_1) + \max\limits_{S\subseteq \mathcal{I}_{H_2}(\mathcal{M})} h_2(S_2)  &~\geq ~  h_{H_1}(O_1) + h_{H_2}(O_2)\\
  &~\geq ~  f(O_1) + f(O_2)\tag*{(monotonicity of $f$)}\\
    &~\geq ~ f(O_1\cup O_2) = f(O^{\star})\tag*{(submodularity of $f$)}.
\end{align*}
\end{proof}
\noindent
Finally, we show that comparison with noisy auxiliary function $\widetilde{f_0}$ causes only a small loss in the approximate ratio.
\begin{claim}
\label{clm:matroid2_comp}
We assume a sufficient small $\varepsilon$ and suppose that $r\in \Omega\left(n^{\frac{1}{3}}\right)$.
Let $S_R$ denote the set returned by Algorithm~\ref{alg:local_search_monotone4}.
With probability at least $1-\frac{2}{n^4}$, we have 
\begin{equation*}
    f(S_R)\geq \left(\frac{1}{2}-O(\varepsilon)\right)\left(f(S_1\cup H_1)+f(S_2\cup H_2)\right).
\end{equation*}
\end{claim}
\begin{proof}
By Lemma~\ref{lmm:matroid1_comp} with $\delta = 1/n^{4}$, we have
\begin{equation}
\label{eq:event_c2}
    \left|\widetilde{f}_0(S)-f_0(S)\right|\leq  \varepsilon f_0(S)
\end{equation}
holds for both $S_1\cup H_1$ and $S_2\cup H_2$ with probability at least $1-2/n^{4}$.
Suppose this is true.
Then for $t\in\{1,2\}$, we have
\begin{align}
    \notag f_0(S_R) &\geq \frac{1}{1+\varepsilon}\cdot\widetilde{f_0}(S_R) && \text{(Ineq.~\eqref{eq:event_c2})}\\
    \notag&\geq \frac{1}{2(1+\varepsilon)}\cdot\left(\widetilde{f_0}(S_1\cup H_1)+\widetilde{f_0}(S_2\cup H_2)\right) && \text{($\widetilde{f_0}(S_R) = \max_{t\in\{1,2\}}\widetilde{f_0}(S_t\cup H_t)$)}\\
    \label{eq:f_0_ineq} &\geq \frac{1}{2}\cdot\frac{1-\varepsilon}{1+\varepsilon}\cdot\left(f_0(S_1\cup H_1)+f_0(S_2\cup H_2)\right). && \text{(Ineq.~\eqref{eq:event_c2})}
\end{align}
We convert this inequality of $f_0$ to that of $f$:
\begin{align*}
    f(S_R) &\geq f_0(S_R) && \text{(Lemma~\ref{lem:bounds_f0})}\\
    &\geq \frac{1}{2}\cdot\frac{1-\varepsilon}{1+\varepsilon}\cdot\left(f_0(S_1\cup H_1)+f_0(S_2\cup H_2)\right) && \text{(Ineq.~\eqref{eq:f_0_ineq})}\\
    &\geq \frac{1}{2}\cdot\frac{1-\varepsilon}{1+\varepsilon}\cdot\left(1-\frac{1}{r}\right)\left(f(S_1\cup H_1)+f(S_2\cup H_2)\right) && \text{(Lemma~\ref{lem:bounds_f0})}\\
    &\geq \left(\frac{1}{2}-2\varepsilon\right)\left(f(S_1\cup H_1)+f(S_2\cup H_2)\right), && \text{($\frac{1-\varepsilon}{1+\varepsilon}\geq 1-2\varepsilon$ and $r\geq 1/\varepsilon$)}
\end{align*}
which completes the proof.
\end{proof}
\noindent
Combining Ineq.~\eqref{eq:f_St_cup_Ht}, Claim~\ref{clm:matroid2_split} and Claim~\ref{clm:matroid2_comp}, we can conclude that with probability $1-O\left(\frac{1}{n^2}\right)$,
\begin{equation*}
 f(S_{i^\star}\cup H_{i^\star})
 \geq \left(\frac{1}{2}\left(1-\frac{1}{e}\right)-O(\varepsilon)\right) f(O^{\star}),
\end{equation*}
which matches Theorem~\ref{thm:matroid2}.

\section{Discussions}
For any cardinality $r\geq 2$, Algorithms~\ref{alg:local_search_monotone} and \ref{alg:local_search_monotone2} return a $\left(\left(1-\frac{1}{r}\right)\left(1-\frac{1}{e}\right)-O\left(\varepsilon\right)\right)$-approximation for monotone submodular maximization problem under noise with query complexity $\mathrm{Poly}\left(n, \frac{1}{\varepsilon}\right)$. When $r\in \Omega\left(\frac{1}{\varepsilon}\right)$, the approximation ratio can be written as $\left(1-\frac{1}{e}\right)-O(\varepsilon)$ (Theorem~\ref{thm:informal_cardi}). 
However, if $r\in O\left(\frac{1}{\varepsilon}\right)$, then this ratio is not close to the optimal one. In this case, \cite{singer2018optimization} provide an algorithm that achieves $\left(1-\frac{1}{r}\right)$-approximation with query complexity $\Omega(n^r)$. It remains open whether there is an algorithm that returns $\left(1-\frac{1}{e}-O(\varepsilon)\right)$-approximation with query complexity $\text{Poly}\left(n,\frac{1}{\varepsilon}\right)$ for $r\in O\left(\frac{1}{\varepsilon}\right)$. The main challenges of our approach to handle this range of $r$ is that the introduction of surrogates (Definition~\ref{def:h_1}) inevitably results in a loss of $\frac{1}{r}$ in approximation ratio. For the special case of $r=1$, one of the few known results is an algorithm by \cite{singer2018optimization} achieving a $\frac{1}{2}$-approximation guarantee in expectation, which is information theoretically tight. To the best of our knowledge, there is no with-high-probability result for this case.

Another limitation of this work is that we only consider independent noise.
For non-i.i.d. noise, \cite{hassidim2017submodular} indicate that no algorithm can achieve a constant approximation when the noise multipliers are arbitrarily correlated across sets.
Considering this, it may be worthwhile to consider special cases of correlated distributions for which optimal guarantees can be obtained. One of them is a model called $d$-correlated noise~\cite{hassidim2017submodular}: a noise distribution is $d$-correlated if for any two sets $S$ and $T$ such that $|S\backslash T|+|T\backslash S| > d$, the noise is applied independently to $S$ and to $T$. The noise multipliers can be arbitrarily correlated when $S$ and $T$ are similar in the sense that $|S\backslash T|+|T\backslash S| \leq d$. We notice that our algorithms can be naturally extended to $d$-correlated noise for $d\in O(1)$. In particular, to adapt Algorithm~\ref{alg:local_search_monotone} to deal with $d$-correlated noise, we need to arbitrarily split $N$ into sets $T_1,\dots,T_{\lfloor\frac{N}{d+1}\rfloor}$ and define the smoothing surrogate function as $h(S) = \frac{1}{L}\sum^L_{l=1}f(S\cup T_l)$ to replace the original surrogate.
\section{Conclusions}
\label{sec:conclusion}
In this work, we study the problem of constrained monotone submodular maximization with noisy oracles.
We design a unified local search framework that allows for inaccuracy in the objective function.
Under this framework, we construct several smoothing surrogate functions to average the noise out.
For cardinality constraints, the local search framework results in algorithms that achieve
a $\left(1-\frac{1}{e}-O(\varepsilon)\right)$-approximation with $\text{Poly}\left(n,\frac{1}{\varepsilon}\right)$ query complexity.
Moreover, for general matroid constraints, the framework obtains an approximation ratio arbitrarily close to $\left(1-\frac{1}{e}\right)/2$, which is the first constant approximation result to our knowledge.

There are many directions in which this work could be extended.
For submodular maximization with noisy oracles under general matroid constraints, there is a gap between the approximation ratio of $\left(1-1/e\right)/2-O(\varepsilon)$ provided in this paper and impossibility results~\cite{feige1998threshold, singer2018optimization}.
The first open question is to close this gap. 
In addition, it would be interesting to consider more complex constraints under noise including knapsack constraints~\cite{chekuri2014submodular}.
It is meaningful to investigate submodular maximization with correlated noise.
Moreover, it is also worthwhile to investigate the robustness of other submodular optimization approaches, e.g., multi-linear extension ~\cite{vondrak2008optimal}.

\section*{Acknowledgements}

We are grateful to Dr. Bei Xiaohui for his good counsel and valuable comments on the manuscript.

\bibliographystyle{plain}
\bibliography{references}

\newpage
\appendix
\section{The Invalid Result for General Matroid Constraints in~\cite{singer2018optimization}}
\label{sec:invalid}
In this section, we recall Algorithm~\ref{alg:wrong_matroid} (\cite[Algorithm 2]{singer2018optimization}) and its analysis Theorem~\ref{thm:wrong_matroid} (\cite[Theorem  4.2]{singer2018optimization}) proposed by~\cite{singer2018optimization} for monotone submodular maximization with noisy oracles under general matroid constraints.
We argue that Algorithm~\ref{alg:wrong_matroid} fails to obtain the approximation guarantee claimed in Theorem~\ref{thm:wrong_matroid}.

For a set $S\subseteq N$, a bundle $\boldsymbol{b}\subseteq N$ and an intersection of matroids $\mathcal{F}$, the \emph{mean value}, \emph{noisy mean value} and \emph{mean marginal contribution} of $\boldsymbol{b}$ given $S$ are, respectively:
\begin{align*}
   & F(S\cup\boldsymbol{b}) := \mathbb{E}_{\boldsymbol{z}\sim\mathcal{B}_S(\boldsymbol{b})}\left[f(S\cup\boldsymbol{z})\right]\\
  & \Tilde{F}(S\cup\boldsymbol{b}) := \mathbb{E}_{\boldsymbol{z}\sim\mathcal{B}_S(\boldsymbol{b})}\left[\Tilde{f}(S\cup\boldsymbol{z})\right]\\
  & F_{S}(\boldsymbol{x}) := \mathbb{E}_{\boldsymbol{z}\sim\mathcal{B}_S(\boldsymbol{b})}\left[f_S(\boldsymbol{b})\right] 
\end{align*}
where $\mathcal{B}_S(\boldsymbol{b})$ represents the ball around $\boldsymbol{b}$, i.e., $\mathcal{B}_S(\boldsymbol{b}) = \{\boldsymbol{b}-x_i+x_j\in \mathcal{F}:x_i\in\boldsymbol{b}, x_j\notin S\cup\boldsymbol{x}\}$.

\begin{algorithm}[htp!]
	\caption{\textsc{SM-Matroid-Greedy}} 
	\label{alg:wrong_matroid}
    \DontPrintSemicolon
    \SetNoFillComment
    \SetKwInOut{Input}{Input}
    \Input{intersection of matroids $\mathcal{F}$, precision $\varepsilon > 0$, $c\leftarrow \frac{56}{\varepsilon}$}
	$S\leftarrow \varnothing$, $X \leftarrow N$\;
	\While{$X \neq S$}{
		$X\leftarrow X\backslash\{\boldsymbol{x}:S\cup \boldsymbol{x}\notin \mathcal{F}\}$\;
		$\boldsymbol{x}\leftarrow \arg\max_{\boldsymbol{b}:|\boldsymbol{b}|=c} \Tilde{F}(S\cup \boldsymbol{b})$\;
		$\hat{\boldsymbol{x}} \leftarrow \arg\max_{\boldsymbol{z}\in \mathcal{B}_S(\boldsymbol{x})}\Tilde{f}(S\cup \boldsymbol{z})$\;
		$S \leftarrow S \cup \hat{\boldsymbol{x}}$\;}
\Return $S$
\end{algorithm}
Algorithm~\ref{alg:wrong_matroid} is a variant of the standard greedy algorithm, which at every iteration adds a bundle $\hat{\boldsymbol{x}}$ of size $c=\Theta(1/\varepsilon)$ instead of a single element.
In each iteration, the algorithm updates the set $X$ of candidate elements (Line 3) in order to obtain a feasible set $S\in \mathcal{F}$.
Then it selects a bundle $\boldsymbol{x}\in\arg\max_{\boldsymbol{b}:|\boldsymbol{b}|=c} \Tilde{F}(S\cup \boldsymbol{b})$ with the largest noisy mean value $\Tilde{F}(S\cup\boldsymbol{b})$ (Line 4).
Finally, Algorithm~\ref{alg:wrong_matroid} evaluates all possible bundles $\boldsymbol{z}\in\mathcal{B}_S(\boldsymbol{x})$ in the ball around $\boldsymbol{x}$ identified at last step and incorporates the one whose noisy value $\Tilde{f}(S\cup \boldsymbol{z})$ is largest (Line 5 and 6).

\cite{singer2018optimization} claim that Algorithm~\ref{alg:wrong_matroid} can achieve the following approximation performance.
\begin{theorem}[\textbf{Wrong theorem~\cite[Theorem 4.2]{singer2018optimization}}]
\label{thm:wrong_matroid}
Let $\mathcal{F}$ denote the intersection of $P\geq 1$ matroids with rank $r\in\Omega\left(\frac{1}{\varepsilon^2}\right)\cap\sqrt{\log n}$ on the ground set $N$, and $f:2^N\rightarrow \mathbb{R}$ be a non-negative monotone submodular function.
Then with probability $1-o(1)$ the \textsc{SM-Matroid-Greedy} algorithm returns a set $S\in\mathcal{F}$ s.t.:
$$f(S)\geq \frac{1-\varepsilon}{P+1}.$$
\end{theorem}
\noindent
Now we argue that Theorem~\ref{thm:wrong_matroid} does not hold even under a single matroid constraint.
\cite{singer2018optimization} employed~\cite[Lemma 3.4]{singer2018optimization} to prove Theorem~\ref{thm:wrong_matroid}.
However, this lemma cannot be generalized beyond cardinality constraints to matroid constraints.
This is because \cite[Lemma 3.4]{singer2018optimization} is based on the following fact: for any bundle $\boldsymbol{b}$ of size $1/\varepsilon$,
\begin{equation*}
    F_S(\boldsymbol{b}) \geq (1-\varepsilon)f_S(\boldsymbol{b}).  \tag*{(\cite[Lemma 2.2]{singer2018optimization})}
\end{equation*} 
%
This fact results in \cite[Corollary 2.3]{singer2018optimization}, which appears as the premise of \cite[Claim 3.1]{singer2018optimization} to prove \cite[Lemma 3.4]{singer2018optimization}.
Under a cardinality constraint, the ball $\mathcal{B}_S(\boldsymbol{b})$ contains all the neighbors of a bundle $\boldsymbol{b}$, which differ from $\boldsymbol{b}$ by only one element.
A bundle $\boldsymbol{b}$ with good margin is likely to be surrounded by neighbors in $\mathcal{B}_S(\boldsymbol{b})$ which also have large margins on average, and the fact that $F_S(\boldsymbol{b}) \geq (1-\varepsilon)f_S(\boldsymbol{b})$ naturally holds. 
However, when we consider a matroid that restricts feasible neighbors of $\boldsymbol{b}$ to be those with small margins only, $F_S(\boldsymbol{b})$ can be far less than $f_S(\boldsymbol{b})$. 
%
%
%
%
%
%
%

We provide a concrete example to show that  Algorithm~\ref{alg:wrong_matroid} cannot obtain any with-high-probability constant approximation under a matroid constraint.
\begin{claim}
\label{clm:counter}
For any constant $m$, there exists a partition matroid $\mathcal{M}$ with rank $r\in \Omega\left(\frac{1}{\varepsilon^2}\right)\cap O(\sqrt{\log n})$ for which Algorithm~\ref{alg:wrong_matroid} fails to achieve an approximate ratio better than $1/m$ with probability at least $1/3$.
\end{claim}
\begin{proof}
Our plan is to construct a partition matroid for which Algorithm~\ref{alg:wrong_matroid} performs badly.
\begin{definition}[\textbf{Partition matroid}]
\label{def:partition}
A matroid $\mathcal{M} = (N, \mathcal{I}(\mathcal{M}))$ is a partition matroid if $N$ is partitioned into $k$ disjoint sets $C_1, C_2, \dots, C_k$ and
$$\mathcal{I}(\mathcal{M}) = \left\{S \subseteq N: |C_i \cap S| \leq d_i \text{\ for\ } i = 1, 2,\dots,k\right\}.$$
\end{definition}
\noindent
Let there be $r$ disjoint sets, where $r\in \Omega\left(\frac{1}{\varepsilon^2}\right)\cap O(\sqrt{\log n})$.
For all $i\in[r-1]$, $C_i$ contains only one element with value $0$.
The set $C_r$ with $|C_r| = n-r+1$ contains most of the elements.
There is a special element $e^\star$ in $C_r$ with $f(e^\star) = m$, and other elements in $C_r$ are all attributed with value $1$.
An independent set of matroid $\mathcal{M}$ contains at most one element from each disjoint set. That is, $d_1 = d_2 = \cdots = d_r = 1$.
The noise distribution will return $m$ with probability $\delta = \frac{1}{2(n-r)}$ and $1$ otherwise.

Now we apply Algorithm~\ref{alg:wrong_matroid} to maximize $f$ under this matroid constraint with a noisy oracle.
Since the elements in $C_i\ (i\in[r-1])$ are all with value 0, we focus on the bundles $\boldsymbol{b}$ that consist of an element in $C_r$ and some other elements $\boldsymbol{b}_{-r}$.
%
%
%
Suppose Algorithm~\ref{alg:wrong_matroid} selects one of them as $\boldsymbol{x}$ at Line 4.
Feasible bundles in $\mathcal{B}_S(\boldsymbol{x})$ containing $e^{\star}$ can only be obtained by changing one of the elements in $\boldsymbol{x}_{-r}$.
Thus there are at most $(c-1)(r-c)$ such bundles.
With probability at least 
$\left(1-\delta\right)^{(c-1)(r-c)}
$,
the noise multipliers on such bundles are all $1$.

On the other hand, at least $n-r$ bundles in $\mathcal{B}_S(\boldsymbol{x})$ do not include $e^\star$.
With probability at least $1-(1-\delta)^{n-r}$, there exists one of these bundles with noise multiplier $m$.
Thus Algorithm~\ref{alg:wrong_matroid} selects a bundle that does not include $e^\star$ at Line 5 with probability at least
$$\left(1-\delta\right)^{(c-1)(r-c)}\left[ 1-\left(1-\delta\right)^{n-r}\right] \geq \left(1-\frac{(c-1)(r-c)}{2(n-r)}\right)\left(1-\frac{1}{\sqrt{e}}\right).$$
When $n$ is sufficiently large, the probability above is at least $1/3$.
\end{proof}

\noindent
By the proof of Claim~\ref{clm:counter}, we have an intuition that the problem of submodular maximization under a partition matroid with any rank may degenerate to that under $1$-cardinality constraint.
%
%
The case of $r=1$ is rather difficult for Problem~\ref{problem:noisy_submodular}, and the only known result is an algorithm by~\cite{singer2018optimization} which achieves a $1/2$-approximation guarantee in expectation for this case. 
To the best of our knowledge, there is no with-high-probability result before.

\section{More Related Work}
Besides cardinality and single matroid constraints, more complex constraints have also been considered in the context of submodular optimization before.
\cite{fisher1978analysis} present an algorithm achieving a $1/(k+1)$-approximation for monotone submodular maximization under $k$ matroid constraints.
\cite{lee2010submodular} subsequently improve the approximation guarantee to $1/(k+\varepsilon)$ for $k>2$.
\cite{sviridenko2004note} provides a $\left(1-\frac{1}{e}\right)$-approximation under knapsack constraints.
With the multilinear relaxation technique, \cite{chekuri2014submodular} obtain a
$0.38/k$-approximation for maximizing a monotone submodular function subject to $k$ matroids
and a constant number of knapsack constraints.

\section{Approximation algorithm for strongly base-orderable matroid constraints}
\label{app:matroid3}
In this section, we consider a special family of matroids, called \emph{strongly base-orderable matroids}.
\begin{definition}[\textbf{Strongly base-orderable matroid}]
\label{def:strongly}
A matroid $\mathcal{M}$ is strongly base-orderable if given any two bases $B_1$ and $B_2$, there is a
bijection $\sigma: B_1 \to B_2$ such that for any $X \subseteq B_1$, $(B_1 \setminus X) \cup \sigma(X)$ is a basis, and $(B_2 \setminus \sigma(X)) \cup X$ is a basis.
\end{definition}
\noindent
As is evident from Definition~\ref{def:strongly}, the cardinality constraint which we discuss in Section~\ref{sec:cardinality} is a special case of strongly base-orderable matroid. 
Moreover, this family of matroids includes many typical matroids as well, such as partition matroids (Definition~\ref{def:partition}) and transversal matroids.
%

Next we present an algorithm (Algorithm~\ref{alg:local_search_monotone5}) and its analysis (Theorem~\ref{thm:matroid3}) that achieves near-tight approximation guarantees subject to strongly base-orderable matroids with rank $r\in \Omega\left(n^{\frac{1}{3}}\right)$.

\begin{theorem}[\textbf{Algorithmic results for cardinality constraints when  $r\in  \Omega\left(n^{\frac{1}{3}}\right) $}]
\label{thm:matroid3}
    Let $\varepsilon>0$ and assume $n \in \Omega\left(\frac{1}{\varepsilon^4}\right)$ is sufficiently large. 
    For any $r\in  \Omega\left(n^{\frac{1}{3}}\right)$, there exists an algorithm that returns a $\left(1-\frac{1}{e}-O(\varepsilon)\right)$-approximation for  Problem~\ref{problem:noisy_submodular} under a strongly base-orderable matroid constraint $\mathcal{I}(\mathcal{M})$, with probability at least  $1-O\left(\frac{1}{n}\right)$ and query complexity at most  $O(n^7\varepsilon^{-1})$ to $\Tilde{f}$.
\end{theorem}
\noindent
Compared with Theorem \ref{thm:matroid2}, Theorem~\ref{thm:matroid3} improves the approximate ratio to $\left(1-\frac{1}{e}-O(\varepsilon)\right)$ for the strongly base-orderable matroids, which have stronger exchangeable structures.

\subsection{Algorithm for Theorem~\ref{thm:matroid3}}

Similar to Algorithm~\ref{alg:local_search_monotone4}, Algorithm~\ref{alg:local_search_monotone5} also contains two phases: a local search phase (Lines 3-5) and a comparison phase (Line 6).
The local search procedure is based on the smoothing surrogate function $h_H$ in Definition \ref{def:h_2}, while the comparison auxiliary function $f_0$ in Definition~\ref{def:f_0} is used at Line 6.
Algorithm~\ref{alg:local_search_monotone4} differs from Algorithm \ref{alg:local_search_monotone5} by running the local search procedure (Algorithm~\ref{alg:local_search}) $\lfloor\frac{r}{l}\rfloor$ times rather than twice, with different smoothing surrogate functions $h_{H_t}$.
\begin{algorithm}[htp!] 
	\caption{Noisy local search under strongly base-orderable matroid constraints} 
	\label{alg:local_search_monotone5}
	\DontPrintSemicolon
    \SetNoFillComment
    \SetKwInOut{Input}{Input}
	\Input{a value oracle to $\Tilde{f}$, rank $r\in  \Omega\left(n^{\frac{1}{3}}\right)$, and $\varepsilon\in (0,1/2)$}
	Let $l\leftarrow \lceil 3\ln n \rceil$\;
     Arbitrarily select a basis $B_0$ and arbitrarily spilt $B_0$ into $\lfloor\frac{r}{l}\rfloor$ parts $H_1,\cdots,H_{\lfloor\frac{r}{l}\rfloor}$ with size $l$ or $l+1$ \;
	\For{$t = 1,\cdots, \lfloor\frac{r}{l}\rfloor$} {
		Let $\widehat{\varphi_{h_{H_t}}}$ be a $(\alpha = \frac{\varepsilon}{4r\ln r},\delta = \frac{3}{n^6},\mathcal{I}_{H_t}(\mathcal{M}))$-approximation of $\varphi_{h_{H_t}}$ as in Lemma \ref{lem:hat_varphi_h4} \;
		$S_t\leftarrow \mathtt{NLS}\left(\widehat{\varphi_{h_{H_t}}},\mathcal{I}_{H_t}(\mathcal{M}),\Delta = \frac{\varepsilon}{4r\ln r}\right)$ \Comment*[r]{Local search phase}
		}
	Let $i^\star \leftarrow \mathop{\arg\max}\limits_{t \in [\lfloor\frac{r}{k}\rfloor]}\widetilde{f_0}(S_t\cup H_t)$ \Comment*[r]{Comparison phase}
	\Return $S_{i^\star}\cup H_{i^\star}$ 
\end{algorithm}

\subsection{Proof of Theorem \ref{thm:matroid3}}
We analyze the query complexity and approximation performance of Algorithm~\ref{alg:local_search_monotone5} to prove Theorem \ref{thm:matroid3}.
To simplify notation, we use $h_t$ to stand for $h_{H_t}$, $\widehat{\varphi_{h_{t}}}$ for  $\widehat{\varphi_{h_{H_t}}}$, and $\mathcal{I}_{t}(\mathcal{M})$ for $\mathcal{I}_{H_t}(r)$ in the analysis.

\paragraph{Query Complexity of Algorithm~\ref{alg:local_search_monotone5}.} 
Similar to the proof of Theorem \ref{thm:matroid2}, 
%
%
%
%
%
Algorithm~\ref{alg:local_search_monotone5} runs the local search procedure $\lfloor r/l\rfloor$ times, and thus the query complexity of Algorithm~\ref{alg:local_search_monotone5} is at most $O({n^7}{\varepsilon}^{-1})$.

\paragraph{Approximation performance analysis of Algorithm~\ref{alg:local_search_monotone5}.}
Similarly, from mononicity of $f$ and Corollary \ref{cor:main_nls}, we have that
\begin{equation}
\label{eq:f_St_cup_base}
f(S_t\cup H_t) \geq \left(1-\frac{1}{e}-\varepsilon\right)\max\limits_{S\subseteq \mathcal{I}_{t}(\mathcal{M})} h_{t}(S)
\end{equation}
holds for all $t\in[\lfloor r/l\rfloor]$ with probability $1-O\left(\frac{r}{n^2}\right)$.
Let $O^{\star} \in \mathop{\arg\max}\limits_{S\in \mathcal{I}(\mathcal{M})}f(S)$ be an optimal solution. 
The following lemma relates $\max\limits_{S\subseteq \mathcal{I}_{t}(\mathcal{M})} h_{t}(S)$ to $f(O^{\star})$.
\begin{lemma}
\label{clm:app_matroid3_split} We have
$$ \mathop{\mathbb{E}}\limits_{t\sim \mathcal{U}}\left[\max\limits_{S\subseteq \mathcal{I}_t(\mathcal{M})} h_t(S)\right]\geq \left(1-\frac{l+1}{r}\right)  f(O^{\star}),$$ where $\mathcal{U}$ is a uniform distribution over $[\lfloor\frac{r}{l}\rfloor]$.
\end{lemma}
\begin{proof}
By Definition~\ref{def:strongly}, there is a bijection $\sigma:B_0\to O^{\star}$ between the elements in $B_0$ and $O^{\star}$ such that for all $t\in[\lfloor\frac{r}{l}\rfloor]$, $(O^{\star}\setminus \sigma(H_t))\in \mathcal{I}_t(\mathcal{M})$.
%
%
The following lemma gives a lower bound of $f(O^*\setminus\sigma(H_t))$ in expectation.
\begin{lemma}
\label{lmm:matroid_split}
Given an arbitrary partition $O_1,\cdots,O_{\lfloor r/l \rfloor}$ of $O^{\star}$ such that $|O_t|\in \{l,l+1\}\ (t\in[\lfloor\frac{r}{l}\rfloor])$, 
we have
\[\mathop{\mathbb{E}}\limits_{t\sim \mathcal{U}}\left[f(O^{\star}\setminus O_t)\right]\geq \left(1-\frac{l+1}{r}\right)f(O^{\star}),\]
where $\mathcal{U}$ is a uniform distribution over $[\lfloor\frac{r}{l}\rfloor]$.
\end{lemma}
\begin{proof}
We index $O^{\star}$ as $\{o_1,\cdots,o_r\}$ and denote by $I_t$ the indexes of the elements in $O_t$ for $t\in[\lfloor\frac{r}{l}\rfloor]$.
We can decompose $f(O^{\star})$ as 
\[f(O^{\star}) = \sum\limits_{i=1}^rf(o_i\mid o_k,k\in[i-1]).\]
$f(O^{\star}\setminus O_t)$ can also be decomposed as below:
\begin{align*}
f(O^{\star}\setminus O_t)
&=  \sum\limits_{i\in [r]\setminus I_t}f(o_i\mid o_k, k\in [i-1]\mbox{ and }k\not \in I_t) \\ 
&\geq \sum\limits_{i\in [r]\setminus I_t}f(o_i\mid o_k, k\in [i-1]).\tag*{(submodularity of $f$)}\\
\end{align*}
Taking an expectation over $t\in[\lfloor\frac{r}{l}\rfloor]$ gives
\begin{align*}
\mathop{\mathbb{E}}_{t\sim \mathcal{U}}\left[f(O^{\star}\setminus O_t)\right]
&\geq \mathop{\mathbb{E}}_{t\sim \mathcal{U}}\left[ \sum\limits_{i\in [r]\setminus I_t}f(o_i\mid o_k, k\in [i-1])\right]\\
&=  \left(1-\frac{1}{\left\lfloor\frac{r}{l}\right\rfloor}\right)  \sum\limits_{i=1}^r f(o_i\mid o_k, k\in [i-1])\geq  \left(1-\frac{l+1}{r}\right)  f(O^{\star}).\\
\end{align*}
\end{proof}
With Lemma~\ref{lmm:matroid_split}, we can prove Lemma~\ref{clm:app_matroid3_split} since
\begin{align*}
    \mathop{\mathbb{E}}\limits_{t\sim \mathcal{U}}\left[\max\limits_{S\subseteq \mathcal{I}_t(\mathcal{M})} h_t(S)\right]
    &\geq   \mathop{\mathbb{E}}\limits_{t\sim \mathcal{U}}\left[h_t(O^{\star}\setminus\sigma(H_t))\right] && \text{($(O^{\star}\setminus \sigma(H_t))\in \mathcal{I}_t(\mathcal{M})$)}\\
     &\geq   \mathop{\mathbb{E}}\limits_{t\sim \mathcal{U}}\left[f(O^{\star}\setminus\sigma(H_t))\right] && \text{(submodularity of $f$)}\\
    &\geq \left(1-\frac{l+1}{r}\right)  f(O^{\star}). && \text{(Lemma~\ref{lmm:matroid_split})}
    \label{eq:card2mainproof3}
\end{align*}
\end{proof}
\noindent
Finally, similar to Lemma~\ref{clm:matroid2_comp}, the following lemma shows that comparison with $\widetilde{f_0}$ causes only a small loss in the approximate ratio.
\begin{lemma}
\label{clm:app_matroid3_comp}
We assume a sufficient small $\varepsilon$ and suppose that $r\in \Omega\left(n^{\frac{1}{3}}\right)$.
With probability at least $1-\frac{2r}{n^4}$, we have 
\begin{equation*}
    f(S_{{i^\star}}\cup H_{{i^\star}})\geq \left(1-O(\varepsilon)\right)\mathop{\mathbb{E}}\limits_{t\sim \mathcal{U}}\left[f(S_t\cup H_t)\right].
\end{equation*}
\end{lemma}
\noindent
The proof idea of this lemma is the same as that of Lemma \ref{clm:matroid2_comp}. 
The only difference is that when proving Lemma~\ref{clm:app_matroid3_comp}, we need to take a union bound of probability that $\widetilde{f_0}(S_t\cup H_t)$ is close to $f(S_t\cup H_t)$ for all $i\in [\lfloor r/l\rfloor]$, thus the success probability in Lemma~\ref{clm:app_matroid3_comp} is at least $1-\frac{2r}{n^4}$.

Now we arrive that with probability $1-O\left(\frac{1}{n}\right)$, we have
\begin{align*}
 f(S_{{i^\star}}\cup H_{{i^\star}})
 \geq ~& \left(1-O(\varepsilon)\right)\mathop{\mathbb{E}}\limits_{t\sim \mathcal{U}}\left[f(S_t\cup H_t)\right]\tag*{(Lemma \ref{clm:app_matroid3_comp})}
\\
 \geq ~& \left(1-\frac{1}{e}-O(\varepsilon)\right)\mathop{\mathbb{E}}\limits_{t\sim \mathcal{U}}\left[\max\limits_{S\subseteq \mathcal{I}_t(\mathcal{M})} h_t(S)\right]  \tag*{(Ineq.~\eqref{eq:f_St_cup_base})}
\\
 \geq ~& \left(1-\frac{1}{e}-O(\varepsilon)\right) f(O^{\star}), 
 \tag*{(Lemma \ref{clm:app_matroid3_split} )}\\
\end{align*}
which matches the approximation performance of Algorithm~\ref{alg:local_search_monotone5} in Theorem~\ref{thm:matroid3}.

\end{document}